\documentclass[journal]{IEEEtran}
\usepackage[utf8]{inputenc}

\usepackage{xcolor}
\usepackage{color, soul}
\usepackage{enumerate}
\usepackage[shortlabels]{enumitem}
\usepackage{url}

\usepackage{hyperref}
\usepackage{cite}
\usepackage{bm, bbm} 
\usepackage{mathtools}
\usepackage{amsmath, amssymb, amsthm}
\usepackage{nicefrac}
\usepackage{empheq}

\usepackage{booktabs}
\usepackage{graphicx}
\usepackage{multirow}
\usepackage{units}
\usepackage{siunitx}
\usepackage{enumitem}

\ifCLASSOPTIONcompsoc
\usepackage[caption=false, font=normalsize, labelfont=sf, textfont=sf]{subfig}
\else
\usepackage[caption=false, font=footnotesize]{subfig}
\fi

\usepackage{hyperref}
\usepackage{svg}

\usepackage{amsmath,amsfonts,amsthm,amssymb}
\usepackage[bb=boondox]{mathalfa}
\usepackage[algoruled]{algorithm2e}
\usepackage{array}

\usepackage[noabbrev,capitalize]{cleveref}

\crefname{equation}{}{}
\Crefname{equation}{}{}

\theoremstyle{definition} 
\newtheorem{proposition}{Proposition}
\newtheorem{theorem}{Theorem}
\newtheorem{definition}{Definition}

\theoremstyle{plain} 

\theoremstyle{remark} 

\usepackage{tikz}

\newenvironment{ldescription}[1]
  {\begin{list}{}%
   {\renewcommand\makelabel[1]{##1\hfill}%
   \settowidth\labelwidth{\makelabel{#1}}%
   \setlength\leftmargin{\labelwidth}
   \addtolength\leftmargin{\labelsep}}}
  {\end{list}}

\newcommand{\set}[1]{\mathcal{#1}} 

\newcommand{\RNum}[1]{\uppercase\expandafter{\romannumeral #1\relax}}
\newcommand{\rNum}[1]{\expandafter\romannumeral #1}

\DeclareMathOperator{\Var}{Var}

\makeatletter
\newcommand{\pushright}[1]{\ifmeasuring@#1\else\omit\hfill$\displaystyle#1$\fi\ignorespaces}
\newcommand{\pushleft}[1]{\ifmeasuring@#1\else\omit$\displaystyle#1$\hfill\fi\ignorespaces}
\makeatother

\newcolumntype{C}[1]{>{\centering\arraybackslash}p{#1}} 

\begin{document}

\title{Inertia Pricing in Stochastic Electricity Markets}

\author{Zhirui~Liang,~\IEEEmembership{Student Member,~IEEE,}
        Robert~Mieth,~\IEEEmembership{Member,~IEEE,}
        and~Yury~Dvorkin,~\IEEEmembership{Member,~IEEE \vspace{-4mm}}



}


\maketitle

\begin{abstract}
Maintaining the stability of renewable-dominant power systems requires the procurement of virtual inertia services from non-synchronous resources (e.g., batteries, wind turbines) in addition to inertia traditionally provided by synchronous resources (e.g., thermal generators). However, the pricing of inertia provision has not been studied in a stochastic electricity market, where the uncertainty characteristics of renewable energy sources (RES) are considered. To fill in this research gap, this paper formulates a chance-constrained stochastic unit commitment model with inertia requirements and computes equilibrium energy, reserve and inertia prices using convex duality. Numerical experiments on an illustrative system and a modified IEEE 118-bus system show the performance of the proposed pricing mechanism. By allowing new virtual inertia providers to contribute to system inertia requirements, the total operating cost reduces. Moreover, the proposed stochastic electricity market internalizes RES uncertainty, which yields additional cost reductions by co-optimizing energy, reserve and inertia procurement.
\end{abstract}


\section*{Nomenclature}

\newcommand{\longestitem}{$C_{0,o,g}$}

\noindent\textit{Sets:}
\begin{ldescription}{\longestitem}
\item [$\mathcal{I}$] Set of nodes 
\item [$\mathcal{N}_i$] Set of nodes that are connected to node $i$
\item [$\mathcal{T}$] Set of time steps in the planing horizon
\end{ldescription}

\vspace{0.5em}
\noindent\textit{Symbols:}
\begin{ldescription}{\longestitem}
\item [ES$_i$] Energy storage at node $i$
\item [G$_i$] Generator at node $i$
\item [W$_i$] Wind farm at node $i$
\end{ldescription}

\vspace{0.5em}
\noindent\textit{Variables:}
\begin{ldescription}{\longestitem}
\item [$e_{i,t}$] Energy level of ES$_i$ at time $t$ (in MWh)
\item [$u_{gi,t}$] On/off status of G$_i$ at time $t$ ($u_{gi,t}=1$ means G$_i$ is on while $u_{gi,t}=0$ means G$_i$ is off)
\item [$u_{gi,t}^{*}$] Optimal value of $u_{gi,t}$
\item [$C_{ES}$] Expected operating cost of all ES (in \$)
\item [$C_G$] Expected operating cost of all generators (in \$)
\item [$H_{ei,t}$] Virtual inertia constant of ES$_i$ at time $t$ (in s)
\item [$H_{eq}$]  Equivalent inertia constant of a system (in s)
\item [$\bm{H}_{wi,t}$] Random virtual inertia constant of W$_i$ at time $t$ (in s)
\item [${H}_{wi,t}$] Deterministic forecast virtual inertia constant of W$_i$ at time $t$ based on $P_{wi,t}$ (in s)
\item [$\bm{P}_{gi,t}$, $\bm{P}_{di,t}$, $\bm{P}_{ci,t}$] Random active power of G$_i$ output, ES$_i$ discharging, and ES$_i$ charging at time $t$ scheduled based on $\bm{P}_{wi,t}$ (in MW)
\item [${P}_{gi,t}$, ${P}_{di,t}$, ${P}_{ci,t}$] Active power of G$_i$ output, ES$_i$ discharging, and ES$_i$ charging at time $t$ scheduled based on $P_{wi,t}$ (in MW)
\item [$\bm{P}_{wi,t}$] Random output power of W$_i$ at time $t$ (in MW)
\item [$P_{wi,t}$] Deterministic forecast power of W$_i$ at time $t$ (in MW)
\item [$\alpha_{gi,t}, \alpha_{di,t}, \alpha_{ci,t}$] Balancing participation factor of G$_i$, ES$_i$ discharging, and ES$_i$ charging at time $t$
\item [$\theta_{i,t}$] Voltage angle of node $i$ at time $t$ (in rad)
\item [$\bm{\omega}_{ht}$] Vector of random nodal forecast error of $H_{wi,t}$ [$\bm{\omega}_{hi,t}, i \in \set{I}$] (in s)
\item [$\bm{\omega}_{pt}$] Vector of random nodal forecast error of $P_{wi,t}$ [$\bm{\omega}_{pi,t}, i \in \set{I}$] (in MW)
\item [$\Delta P_{gi}$] Inertial response of G$_i$ (in MW)
\item [$\bm{\Omega}_{ht}$] System-wide wind inertia forecast error at time $t$ ($\bm{\Omega}_{ht}=\sum\nolimits_{i \in \set{I}} {\bm{\omega}_{hi,t}}=e^\mathrm{T} \bm{\omega}_{ht}$)
\item [$\bm{\Omega}_{pt}$] System-wide wind power forecast error at time $t$ ($\bm{\Omega}_{pt}=\sum\nolimits_{i \in \set{I}} {\bm{\omega}_{pi,t}}=e^\mathrm{T} \bm{\omega}_{pt}$)
\end{ldescription}

\vspace{0.5em}
\noindent\textit{Parameters:}
\begin{ldescription}{\longestitem}
\item [$c_{0i}$] Constant term in the cost function of G$_i$ (in \$)
\item [$c_{1i}$] First-order cost coefficient of G$_i$ (in \$/MWh)
\item [$c_{2i}$] Second-order cost coefficient of G$_i$ (in \$/MWh$^2$)
\item [$c_{ci}$, $c_{di}$] First-order cost coefficient of ES$_i$ charging and discharging  (in \$/MWh)
\item [$e$] Vector of ones of appropriate dimensions
\item [$f_0$] Reference system frequency (\unit[50 or 60]{Hz})
\item [$f_{\max}^{\prime}$] Maximum admissible rate of change of frequency (RoCoF) (in Hz/s)
\item [$k_i$] Charging and discharging efficiency of ES$_i$
\item [$m_{bi}$]  Whole rotor mass of wind turbine at node $i$ (in kg)
\item [$r_{bi}$] Effective rotor radius of wind turbine at node $i$ (in m)
\item [$t_m$] Time instance of frequency nadir (in s)
\item [$B_{i,j}$] Susceptance of the line between node $i$ and $j$ (p.u.)
\item [$E^k_{bi}$] Kinetic energy stored in the rotating mass of the wind turbine at node $i$ (in kg$\cdot$m$^2$/s)
\item [$E_i^{\min}$, $E_i^{\max}$] Lower and upper limit of $e_{i,t}$ (in MWh)
\item [$H_{bi}$] Virtual inertia constant of wind turbine at node $i$ (in s)
\item [$H_{ei,t}^{\max}$] Upper limit of $H_{ei,t}$ (in s)
\item [$H_{gi}$] Inertia constant of G$_i$ (in s)
\item [$H_{\min}$] Minimum equivalent inertia requirement (in s)
\item [$J_{bi}$] Moment of inertia of wind turbine at node $i$ (in kg$\cdot$m$^2$)
\item [$N_{wi}$] Number of wind turbines at node $i$
\item [$P_{bi}^{\max}$] Rated power of wind turbine at node $i$ (in MW)
\item [$P_{ei}^{\max}$] Rated power of ES$_i$ (in MW)
\item [$P_{gi}^{\max}$] Rated power of G$_i$ (in MW)
\item [$P_{im}^{\max}$] Maximum anticipated power imbalance (in MW)
\item [$P_{sys}$] Total installed generation capacity in a system (in MW)
\item [$S_{i,j}$] Thermal capacity of line between node $i$ and $j$ (p.u.)
\item [$T$] Time constants for all synchronous machines (in s)

\item [$\epsilon_{gi}$, $\epsilon_{di}$, $\epsilon_{ci}$] Probability of power constraint violations for $\{\bm{P}_{gi,t}\}_{t\in \set{T}}$, $\{\bm{P}_{di,t}\}_{t\in \set{T}}$, $\{\bm{P}_{ci,t}\}_{t\in \set{T}}$
\item [$\theta_{ref,t}$] Voltage angle of the reference node at time $t$ (in rad)
\item [$\mu_{hi,t}$, $\mu_{pi,t}$] Mean of the distribution of $\bm{\omega}_{hi,t}$ and $\bm{\omega}_{pi,t}$
\item [$\sigma_{hi,t}^2$, $\sigma_{pt}^2$] Variance of the distribution of $\bm{\omega}_{hi,t}$ and $\bm{\omega}_{pi,t}$
\item [$\phi_{bi}$] Rotor speed of wind turbine at node $i$ (in rad/s)
\item [$\Delta {f_{\max}}$] Maximum admissible frequency deviation at frequency nadir (in Hz)
\item [$\mathrm{M}_{ht}$, $\mathrm{M}_{pt}$] Mean of the distribution of $\bm{\Omega}_{ht}$ and $\bm{\Omega}_{pt}$
\item [$\Sigma_{ht}^2$, $\Sigma_{pt}^2$] Variance of the distribution of $\bm{\Omega}_{ht}$ and $\bm{\Omega}_{pt}$

\end{ldescription}

\IEEEpeerreviewmaketitle

\section{Introduction}
\label{sec:introduction}

The massive deployment of renewable energy sources (RES) is a cornerstone to achieving emission-reduction goals. 
For example, the U.S. Biden-Harris administration recently approved the \unit[800]{MW} Vineyard Wind energy project as a first step towards a total of \unit[30]{GW} new off-shore wind generation \cite{offshore_wind_project}. 
However, increasing injections from stochastic RES amplify uncertainties in power system and electricity market operations \cite{dvorkin2015uncertainty} and require the procurement of sufficient balancing capabilities to ensure reliable electricity delivery.  
Additionally, RES and complementary modern energy resources (e.g., battery systems) are interfaced via power-electronic converters and therefore -- in contrast to traditional synchronous generators -- do not naturally contribute to system \textit{inertia requirements} \cite{datta2020battery}. For example, replacing a conventional \unit[800]{MW} generator whose inertia constant is 5 seconds with a wind farm of the equivalent capacity reduces the system inertia by \unit[4]{GWs}, which is observable even in large systems~\cite{Inertia_and_Power_Grid}. Therefore, power systems with a high penetration of inverter-interfaced RES require new frequency control technologies that replenish the inertial response of retired conventional power plants with spinning synchronous generators for sudden power imbalance. Such technologies are commonly referred to as virtual (or hidden, emulated, synthetic) inertia \cite{fernandez2020review}.

The importance of virtual inertia provision and pricing has been studied in recent literature, e.g, in \cite{poolla2020market,badesa2020pricing,paturet2020economic}. In \cite{poolla2020market}, the authors develop a inertia pricing mechanism in the liberalized energy market based on the Vickrey-Clarke-Groves (VCG) payment rule, in which market participants bid to provide inertia. While \cite{poolla2020market} explicitly allows an economically efficient decentralized inertia allocation, this approach may leave the market operator with an inefficient surplus and does not capture the interdependencies with energy and reserve provision.
Alternatively, \cite{badesa2020pricing} derives a joint marginal pricing scheme for inertia and multi-speed frequency response services (including first-order and second-order frequency regulations) based on a unit commitment model with multiple frequency-security constraints, i.e. controlling the rate of change of frequency (RoCoF), frequency nadir and frequency quasi-steady-state (q-s-s) within acceptable limits. 
Similarly, \cite{paturet2020economic} studies the procurement and pricing of inertia using a frequency-constrained unit commitment formulation while only RoCoF constraint is considered, and it designs and compares three pricing and payment schemes to ensure that all the inertia service providers receive non-negative profits. However, while \cite{poolla2020market,badesa2020pricing,paturet2020economic} internalize inertia services into market-clearing mechanisms, they do not account for the stochastic characteristic of RES injections, which is an indispensable attribute of RES-rich power systems. 

RES stochasticity comprises variability and uncertainty. The variability of RES is the random very-short term fluctuation of RES generation caused by physical processes in the atmosphere, while uncertainty of RES results captures forecast errors \cite{dvorkin2015uncertainty}. 
The uncertainty of RES increases the possibility that the system cannot meet the load requirement in real-time scheduling, which requires suitable balancing regulation and flexible reserve capacity. Moreover, stochastic RES injections may cause high price volatility \cite{ela2017electricity}. 
Hence, different uncertainty modeling techniques for electricity pricing under uncertainty have been developed \cite{aien2016comprehensive}. 

Most previously proposed stochastic electricity market designs rely on \textit{scenario-based} stochastic programming \cite{wong2007pricing,pritchard2010single,morales2012pricing,kazempour2018stochastic}. 
However, besides high computational requirements that limit the number of scenarios that can be considered \cite{dupavcova2003scenario}, the accuracy of the scenario-based method highly depends on how well the chosen scenarios can capture both the range and correlation structures of uncertain parameters. 
Meanwhile, the scenario-based market clearing approaches are usually unable to be welfare-optimal, revenue adequate and cost recovering both in expectation and in each scenario.
Alternatively, \textit{chance constraints}, which rely on computationally tractable risk metrics to internalize RES uncertainty, have been proposed as a promising candidate for practical stochastic electricity markets.
The work in \cite{kuang2018pricing,dvorkin2019chance,mieth2020risk,ratha2019exploring} showed that wholesale electricity market designs can efficiently internalize the uncertainty of renewable generation resources and the reliability requirements of the system operator in the price formation process using chance constraints. However, the approaches in \cite{kuang2018pricing,dvorkin2019chance,mieth2020risk,ratha2019exploring} do not consider inertia services and requirements, which are important for RES-rich power systems because of their influence on commitment and dispatch decisions and, thus, on the resulting prices \cite{paturet2020economic}.

Simultaneous market-clearing procedures for energy and reserve are widely used in the U.S., e.g., by New York Independent System Operator (NYISO) \cite{nyiso2021manual11}, and their advantages of achieving a greater value of social welfare relative to sequential markets have been demonstrated, e.g., in \cite{galiana2005scheduling}. Similarly to energy and reserve, inertia provision is strongly coupled with the other two services and depends on the commitment status of generators. 
Therefore, similar benefits are expected from simultaneously clearing of energy, reserve and inertia, as shown in \cite{doherty2005frequency,davarinejad2017incorporating}. However, \cite{doherty2005frequency,davarinejad2017incorporating} only studied the inertia from conventional generators and deterministic markets, while the virtual inertia and the uncertainty of RES are not considered.
To co-optimize the provision of energy, reserve and inertia, (including virtual inertia) services in renewable-rich power systems, this paper formulates a chance-constrained unit commitment (CC-UC) problem that internalizes (i) RES uncertainty and resulting reserve requirements and (ii) inertia requirements.
The proposed CC-UC moderately modifies established unit commitment problems and can be solved efficiently. We also show that efficient prices (i.e., prices that support a cost-minimizing competitive equilibrium) for all three products can be obtained from the proposed CC-UC, thus leading to an inertia-aware stochastic electricity market design.

The contributions of this work are: 
\begin{itemize}
    \item In contrast with \cite{poolla2020market,badesa2020pricing,paturet2020economic, kuang2018pricing,dvorkin2019chance,mieth2020risk,ratha2019exploring}, this paper co-optimizes the procurement of energy, reserve and inertia providing services in a RES-rich power system and the pricing of these services in a centralized stochastic electricity market. The market design is based on the chance-constrained unit commitment formulation which is recast as a mixed-integer quadratic program (MIQP) as \cite{kuang2018pricing,dvorkin2019chance} assuming normally distributed random variables.
    \item Similar to the current industry practices, this paper adopts the marginal cost-based pricing principle. It also proves rigorously that the resulting dispatch and pricing decisions are efficient and constitute a competitive equilibrium.
    \item Unlike in \cite{kuang2018pricing,dvorkin2019chance,mieth2020risk,ratha2019exploring,doherty2005frequency,davarinejad2017incorporating}, the proposed market design accommodates energy storage systems (ES, e.g., utility scale battery systems) and enables inertia provision from both ES and utility-scale RES (e.g., wind farms).
\end{itemize}

The proposed pricing mechanism is based on current U.S. practice, where reserves are procured simultaneously with energy. For the current market designs in most European countries, which do not immediately permit a simultaneous clearing of energy and various reserve products, our approach can provide decision support for the system operator to co-optimize energy, reserve and inertia requirements and assign precise spatio-temporal values to these requirements, even if they are not traded in a joined market framework.

The proposed market design will be beneficial for various stakeholders. From the perspective of system operators, the total operating cost in RES-rich systems will be reduced by replacing some expensive synchronous inertia providers with cheaper virtual inertia providers. Consequentially, power customers will benefit from lower energy prices. Further, the potential of ES is more fully exploited by allowing them to provide energy, reserve and inertia services, which bring additional revenue to the ES owners. Similarly, wind farms can generate an additional income by providing virtual inertia to the system. The proposed chance-constrained market clearing with explicit reserve and inertia prices, avoids substantial price fluctuations and offers additional revenue sources. Hence, conventional generators can embrace a RES-rich system with less revenue uncertainty.

\section{Preliminaries}
\label{sec:Preliminaries}
In this section, we will introduce some preliminary models, i.e., RES uncertainty models, reserve allocation policy, inertia constants of different resources, and system inertia requirements. We consider a power system with conventional generators, ES and RES, where both ES and RES are equipped to provide virtual inertia. In this paper, ES refers to utility-scale battery energy storage, because it is currently the most common non-synchronous ES technology.
Less common ES systems that could provide inertia services, e.g., supercapacitors or flywheels \cite{fernandez2020review}, can be added to the proposed model. 
Moreover, our RES models focus on large-scale wind farms and their ability to provide virtual inertia, because wind power is the main source of uncertainty in current transmission-level power systems \cite{bienstock2014chance}.
However, the proposed models can be extended to accommodate uncertainty and inertia provision from other RES technologies (e.g. solar photovoltaic), but this may require adapted virtual inertia models. See \cite{dreidy2017inertia}.

\subsection{Uncertainty Model of Wind Power}
\label{subsec:Uncertainty_Model}

We define $\set{I}$ as the set of nodes in the transmission network indexed by $i$, and $\set{T}$ as the set of time steps in the planing horizon indexed by $t$. In this paper we use \textbf{bold} symbols to indicate random variables. Following \cite{dvorkin2019chance} and \cite{mieth2020risk}, uncertain wind power injection $\bm{P}_{wi,t}$ at node $i$ and time $t$ is modeled as:
\begin{equation}
    \bm{P}_{wi,t} = P_{wi,t} - \bm{\omega} _{pi,t}  \label{Pwit_bold},
\end{equation}
where $\bm{P}_{wi,t}$ is a random variable composed of wind power forecast $P_{wi,t}$ and random forecast error $\bm{\omega}_{pi,t}$. 

Following previous works, e.g., in \cite{bienstock2014chance,kuang2018pricing,dvorkin2019chance,mieth2020risk,ratha2019exploring}, we assume that $\bm{\omega}_{pi,t}$ follows a normal forecast error distribution, i.e., $\bm{\omega}_{pi,t} \sim N(\mu_{pi,t},\sigma_{pi,t} ^2)$, where mean ($\mu_{pi,t}$) and variance ($\sigma_{pi,t}$) may vary over time $t$ and node $i$.
While the effectiveness of Normal distributions to model wind power forecast errors has been demonstrated, e.g., in \cite{dvorkin2015uncertainty}, more general distribution assumptions can be adopted \cite{roald2015security,bienstock2014chance,dvorkin2019chance}. In this paper, we assume uncorrelated forecast errors at the individual wind sites, which holds true in the typical system dispatch intervals (\unit[15--60]{min}) for wind farms that are more than \unit[10]{km} apart \cite{bienstock2014chance}. If empirical data does indicate correlations that can not be neglected, an alternative formulation as shown in \cite{mieth2020risk} can be used.
Finally, note that the distributional characteristics of forecast error $\bm{\omega}_{pi,t}$ for a given forecast $P_{wi,t}$ differ from the distribution of absolute wind power injections over time at a given wind farm, which are often modelled through Weibull distributions \cite{bowden1983weibull}.

\subsection{Real-Time Balancing Regulation}

Compensating forecast error $\bm{\omega}_{pi,t}$ in real-time requires procuring balancing reserves to continuously match power supply and demand. 
Specifically, the burden of balancing regulation is distributed among controllable resources, i.e., generators and ES using balancing participation factors $\alpha_{gi,t},\ \alpha_{di,t},\ \alpha_{ci,t} \in[0,1]$, where $c$ and $d$ in subscripts of variables indicate charging and discharging states of ES, respectively.
These participation factors are modeled as decision variables and capture the relative amount of system-wide forecast error $\bm{\Omega}_{pt}$ that a resource at node $i$ and time $t$ must balance. 
We define the system-wide wind power forecast error $\bm{\Omega}_{pt}=\sum\nolimits_{i \in \set{I}} {\bm{\omega}_{pi,t}} = e^\mathrm{T} \bm{\omega}_{pt}$, where $\bm{\omega}_{pt}$ is a column vector collecting all nodal forecast errors at time $t$, and $e$ is a column vector of ones of appropriate dimensions. Assume that $\{{\bm{\omega}_{pi,t}}\}_{\forall{i}\in\set{I}}$ are jointly normally distributed random variables, then $\bm{\Omega}_{pt}$ also follows a normal distribution, i.e., $\bm{\Omega}_{pt} \sim N(\mathrm{M}_{pt},\Sigma_{pt}^2)$, where $\mathrm{M}_{pt} = \mathbb{E}[\bm{\Omega}_{pt}] = \sum\nolimits_{i \in \set{I}} {\mu_{pi,t}}$, and $\Sigma_{pt}^2 = \Var[\bm{\Omega}_{pt}] = e^\mathrm{T} \mathrm{Cov}[\bm{\omega}_{pt}] e$. (If all $\{{\bm{\omega}_{pi,t}}\}_{\forall{i}\in\set{I}}$ are independent random variables, then $\Sigma_{pt}^2$ can be simplified as $\sum\nolimits_{i \in \set{I}} {\sigma_{pi,t}^2}$.)
Therefore, the real-time active power output of each generator (${\bm{P}_{gi,t}}$) can be modeled as:
\begin{equation}
    \bm{P}_{gi,t} = P_{gi,t} + \alpha _{gi,t} \bm{\Omega}_{pt}. \label{Pgit_bold}
\end{equation}

Note that balancing participation factors $\alpha_{gi,t}$ establish a linear relationship between system imbalance $\Omega_{p,t}$ and the resulting generator response. 
This affine balancing policy resembles droop control strategies employed in both primary and secondary frequency control governed by automatic generator control (AGC) systems \cite{chertkov2017chance, bienstock2014chance,lubin2019chance}. Note that the non-affine control policy proposed in \cite{roald2015optimal} is also implementable within our framework.

Similarly, the real-time discharging and charging power of each ES (${\bm{P}_{di,t}}$ and ${\bm{P}_{ci,t}}$) can be modeled as: 
\begin{align}
    \bm{P}_{di,t} &= P_{di,t} + \alpha _{di,t} \bm{\Omega}_{pt} \label{Pdit_bold}\\
    \bm{P}_{ci,t} &= P_{ci,t} + \alpha _{ci,t} \bm{\Omega}_{pt}. \label{Pcit_bold}
\end{align}
where $\alpha _{gi,t} \bm{\Omega}_{pt}$, $\alpha _{di,t} \bm{\Omega}_{pt}$ and $\alpha _{ci,t} \bm{\Omega}_{pt}$ are the real-time balancing power provided respectively by generator, ES discharging and ES charging at node $i$ and time $t$. Thus, the system is balanced if $\sum\nolimits_{i \in \set{I}} {\left( {{\alpha _{gi,t}} + {\alpha _{di,t}} - {\alpha _{ci,t}}} \right)} = 1 $. If ES are not eligible to participate in balancing reserve provisions, then ${\alpha _{di,t}} = {\alpha _{ci,t}} = 0$. 

\subsection{Equivalent System Inertia with Synchronous and Non-synchronous Providers}
\label{subsec:Equivalent_Inertia}

Traditionally, power system inertia refers to stored kinetic energy in the rotating mass of synchronous generators \cite{kundur2007power}. Each generator is characterized by its inertia constant $H_{gi}$ (in s), and its inertial response $\Delta P_{gi}$ (in MW) is captured as:
\begin{equation}
    \Delta P_{gi} = -\frac{2H_{gi}P_{gi}^{\max}}{f_0} \frac{df}{dt}   \label{Delta P},
\end{equation}
where $f$ is the real-time system frequency (in Hz), ${f_0}$ is the reference system frequency (\unit[50 or 60]{Hz}) and ${P_{gi}^{\max}}$ (in MW) is the rated power of generator at node $i$.

On the other hand, \textit{virtual} inertia is an instant injection or withdrawal of electrical power as a response to frequency deviations from resources that do not naturally vary their power output as a function of system frequency. Instead, these resources require specific control policies that mimic the inertial response of rotating synchronous generators, e.g. a virtual synchronous machine (VSM) algorithm \cite{markovic2018lqr}. Thus, the inertia constant of ES at node $i$ ($H_{ei}$, in s), which is a parameter in the VSM algorithm, can be set based on ES technical limits. That is, the inertial response of ES at node $i$ can be computed as $H_{ei} P_{ei}^{\max}$, where $H_{ei}$ is a decision variable and $P_{ei}^{\max}$ is a parameter denoting the maximum discharging power of ES. 

In turn, wind turbines typically emulate inertia by utilizing the kinetic energy of blade rotation \cite{dreidy2017inertia}. 
Thus, inertia constant $H_{bi}$ (in s) of a wind turbine at node $i$ can be modeled as:
\begin{equation}
    H_{bi} = {E^k_{bi}}/{P_{bi}^{\max}} = ({J_{bi} \phi_{bi}^2})/ ({2P_{bi}^{\max}}),
\end{equation}
where $E^k_{bi}$ is the kinetic energy stored in the rotating mass of the wind turbine at node $i$ (in kg$\cdot$m$^2$/s), $P_{bi}^{\max}$ is its rated power (in MW), $\phi_{bi}$ is its rotor speed (in rad/s) and $J_{bi}$ is its moment of inertia (in kg$\cdot$m$^2$) given by: $J_{bi} = m_{bi} r_{bi}^2/9$,
where $m_{bi}$ is the mass of the whole rotor including the three blades (in kg) and $r_{bi}$ is the effective rotor radius (in m) \cite{morren2006inertial}. 
Assuming that all turbines comprising a wind farm are identical and neglecting wind speed differences within a wind farm, all turbines can provide the same amount of virtual inertia. Thus, if there are $N_{wi}$ wind turbines at node $i$, then the inertia constant of the wind farm at node $i$ is $H_{wi} =(N_{wi}{E^k_{bi}})/ (N_{wi}{P_{bi}^{\max}})=H_{bi}$.

Since the rotor speed $\phi_{bi}$ of wind turbines at each node depends on the random wind speed, $\bm{H}_{wi}$ is also a random and time-varying variable. We assume that forecast error $\bm{\omega}_{hi,t}$ of $\bm{H}_{wi,t}$ is also normally distributed, i.e., $\bm{H}_{wi,t} = H_{wi,t} - \bm{\omega}_{hi,t}$ and $\bm{\omega}_{hi,t}\sim N(\mu_{hi,t},\sigma_{hi,t} ^2)$. 
Similar to $\bm{\Omega}_{pt}$, we define the system-wide wind inertia forecast error $\bm{\Omega}_{ht}=\sum\nolimits_{i \in \set{I}} {\bm{\omega}_{hi,t}}$, and $\bm{\Omega}_{ht}$ also follows a normal distribution, i.e., $\bm{\Omega}_{ht} \sim N(\mathrm{M}_{ht},\Sigma_{ht}^2)$, where $\mathrm{M}_{ht} = \mathbb{E}[\bm{\Omega}_{ht}] = \sum\nolimits_{i \in \set{I}} {\mu_{hi,t}}$, and $\Sigma_{ht}^2 = \Var[\bm{\Omega}_{ht}] = e^\mathrm{T} \mathrm{Cov}[\bm{\omega}_{pt}] e$.
Note that wind power and wind inertia forecast error, $\bm{\omega}_{hi,t}$ and $\bm{\omega}_{pi,t}$, both depend on wind speed. However, $\bm{\omega}_{hi,t}$ and $\bm{\omega}_{pi,t}$ are never part of the same chance constraint in the proposed model and the reformulation in Section~\ref{subsec:Deterministic_Equivalent} only requires each random variable to be captured by a normal distribution and does not make assumptions on dependency.

Since $H_{ei}$ and $\bm{H}_{wi}$ have the same unit (s) as inertia constant $H_{gi}$ of traditional generators, they can be combined directly based on the rated output power of each resource ($P_{gi}^{\max}$, $P_{ei}^{\max}$ and $P_{wi}^{\max}$, in MW). Thus, the equivalent inertia $H_{eq}$ (in s) of a system relying on both traditional and virtual inertia is:
\begin{equation}
    H_{eq,t} {\rm{=}} {\sum\limits_{i \in \set{I}} \mathbb{E}_{\bm{\omega}_{hi,t}}} [H_{gi} P_{gi}^{\max} {\rm{+}} H_{ei} P_{ei}^{\max} {\rm{+}} \bm{H}_{wi,t} P_{wi}^{\max} )]/{P_{sys}},  \label{Heq}
\end{equation}
where the total installed generation capacity in the system (in MW) is given by:
\begin{equation}
    P_{sys}=\sum\nolimits_{i \in \set{I}}  (P_{gi}^{\max} + P_{ei}^{\max} + P_{wi}^{\max}).\label{Psys}
\end{equation}

\subsection{System Inertia Requirements}
\label{subsec:Inertia_Requirements}

\begin{figure}
    \centering
    \includegraphics[width=0.95\linewidth]{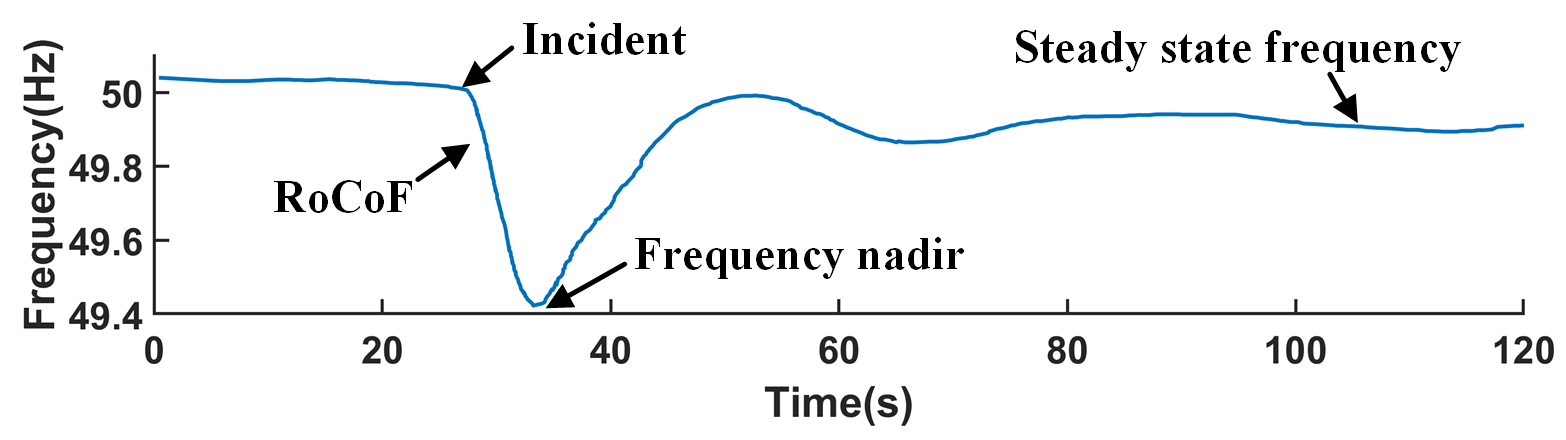}
    \caption{System frequency excursions caused by a sudden generation deficit.}
    \label{fig:RoCoF definition}
\end{figure}

Fig.~\ref{fig:RoCoF definition} shows the system frequency after a sudden generation deficit (e.g., an unplanned generator outage or load spike). 
The resulting frequency excursion can be characterized by the rate of change of frequency (RoCoF), the frequency nadir (the minimum instantaneous frequency) and the steady-state frequency. The greatest RoCoF occurs at the outset of the excursion, when the power imbalance is maximal.
To ensure power system stability, the inertial response must be such that the maximum RoCoF, frequency nadir and steady state frequency are kept within predefined limits. Thus, by varying the equivalent system inertia $H_{eq}$ in \cref{Heq}, one must ensure the provision of the minimum system inertia requirements ($H_{\min}$). 
Note that the maximum RoCoF only depends on $H_{eq}$ and the steady-state frequency deviation is determined only by the damping and droop gain in the affine control process, while the frequency nadir depends on all aforementioned variables \cite{paturet2020economic}. Thus, $H_{\min}$ is a function of the maximum allowable RoCoF and frequency nadir, and it should be equal to the greater value of the inertia constant requirements determined by the RoCoF limit and the frequency nadir limit. 
As reported in \cite{paturet2020economic}, $H_{\min}$ can be computed as:
\begin{align}
    H_{\min} &= \max \bigg\{ ( \overbrace{\vphantom{\bigg\{}\left| P_{im}^{\max} \right| f_0) / (2{f_{\max}^{\prime}} P_{sys})}^{\text{Determined by RoCoF limit}}, \nonumber\\ 
    & \underbrace{T (R_g - F_g){\left( {\frac{\left| P_{im}^{\max} \right|{e^{-\varsigma \omega_n t_m}}} {\Delta f_{\max}(D + R_g)+ \left| P_{im}^{\max} \right|} }\right)^2}}_{\text{Determined by frequency nadir limit}}\bigg\}, \label{Hmin}
\end{align}
where ${\left| {P_{im}^{\max}} \right|}$ is the absolute value of the maximum anticipated power imbalance in the system (e.g. following a credible contingency), $f_{\max}^{\prime}$ is the maximum admissible RoCoF, $\Delta {f_{\max}}$ is the maximum admissible frequency deviation at nadir, $T$ is the time constant of generators (assume the equal time constant for all synchronous generators),  $t_m$ is the time instance of frequency nadir,  $R_g$, $F_g$, $D$, $\varsigma$, and $\omega_n$ are defined as the droop gain, the fraction of power generated by synchronous generators, the damping constant, the damping ratio, and the natural frequency, respectively, and the specific calculation method of these values can be found in \cite{markovic2018lqr}.
The first and the second terms in \cref{Hmin} are determined by the RoCoF and nadir limits, respectively. However, we note similarly to \cite{paturet2020economic} that RoCoF limit violations are more common and greater than nadir limit violations. Therefore, we focus on the inertia requirements determined by the RoCoF limit for calculating $H_{\min}$.

Note that the same requirements must hold for an event with a sudden generation surplus (e.g., unplanned outage of a large industrial load), where the frequency deviation must be contained by an upper limit (frequency zenith).

\section{Chance-Constrained Unit Commitment Model}
\label{sec:models}

In RES-rich power systems, the commitment decisions are affected by balancing reserve and inertia requirements.
Therefore, we now derive a suitable uncertainty- and inertia-aware unit commitment (UC) model, using the RES uncertainty models, balancing regulation policy and inertia requirements shown in Section~\ref{sec:Preliminaries}.
As discussed in Section~\ref{sec:introduction} above, we use chance constraints to internalize the uncertain RES, generator and ES injections, as given by \cref{Pgit_bold,Pdit_bold,Pwit_bold,Pcit_bold}. 
First, we formulate a benchmark model (Section~\ref{subsec:Benchmark_Model}) based on the existing work in \cite{dvorkin2019chance}, which reflects the current state-of-the-art in chance-constrained unit commitment (CC-UC) and market clearing, but which we extended to account for inertia requirements. 
Then, we propose a modification of this model in Section~\ref{subsec:Proposed_Model} to account for virtual inertia provision from ES and wind farms and ES reserve participation.
Finally, we recast the proposed model as a deterministic, computationally tractable convex quadratic program to facilitate price analysis in Section-\ref{sec:prices}.

\subsection{Benchmark Model}
\label{subsec:Benchmark_Model}

The CC-UC model in \cref{model_1} motivated by \cite{dvorkin2019chance} will be used as the benchmark as it considers generators providing energy, reserve and inertia to the system, while ES and RES can only provide energy, but not reserve or virtual inertia. 
\allowdisplaybreaks
\begin{subequations}
\begin{align}
&\min_{\substack{ \{P_{gi,t},\alpha _{gi,t}  P_{di,t}, P_{ci,t},\\ u_{gi,t}  \}_{t \in \set{T}, i \in \set{I}} }}
    \sum\limits_{t \in \set{T}} \sum\limits_{i \in \set{I}} \Big\{ \mathbb{E}_{\bm{\omega}_{pi,t}} [c_{gi}(\bm{P}_{gi,t}, u_{gi,t})] \nonumber\\
    & \hspace{3.5cm} + c_{di}(P_{di,t}) + c_{ci} (P_{ci,t})\Big\}   \label{model_1_a}\\
    &\text{s.t. }  \forall t\in\set{T}: \nonumber \\
    & (\tau_{gi,t}^{+}): \mathbb{P}_{\bm{\omega}_{pi,t}} \!\! \left[ {\bm{P}_{gi,t} \le u_{gi,t} P_{gi}^{\max }} \right] \ge 1 - \epsilon_{gi} \ \ \forall{i}\in\set{I} \label{model_1_b}\\
    & (\tau_{gi,t}^{-}): \mathbb{P}_{\bm{\omega}_{pi,t}} \!\! \left[ { \bm{P}_{gi,t} \ge u_{gi,t} P_{gi} ^{\min}} \right] \ge 1 - \epsilon_{gi} \ \ \forall{i}\in\set{I} \label{model_1_c}\\
    & (\xi_{i,t}^{-} ,\xi_{i,t}^{+}): 0 \le P_{di,t} \le P_{ei}^{\max} \quad \forall{i}\in\set{I} \label{model_1_d}\\
    & (\nu_{i,t}^{-} ,\nu_{i,t}^{+}): 0 \le P_{ci,t} \le P_{ci}^{\max} \quad \forall{i}\in\set{I} \label{model_1_e}\\
    & (\beta _{i,t}^{-} ,\beta _{i,t}^{+} ): E_i^{\min } \le e_{i,t} \le E_i^{\max } \quad \forall{i}\in\set{I} \label{model_1_f}\\
    & (\eta _{i,t}): e_{i,t} = e_{i,t-1} - P_{di,t} / k_i + P_{ci,t} k_i \quad \forall{i}\in\set{I} \label{model_1_g}\\
    & (\rho_{gi,t}^{-} ,\rho_{gi,t}^{+} ): 0 \le \alpha _{gi,t} \le u_{gi,t} \quad \forall{i}\in\set{I} \label{model_1_h}\\
    & (\lambda_t): \sum\nolimits_{i \in \set{I}} {\left( {P_{gi,t} + P_{di,t} - P_{ci,t} + P_{wi,t} -d_{i,t} } \right)}  {\rm{=}}  0   \label{model_1_i}\\
    & (\gamma_t): \sum\nolimits_{i \in \set{I}} {\alpha _{gi,t}}  = 1  \label{model_1_j}\\
    & (\chi _t):\sum\nolimits_{i \in \set{I}} {u_{gi,t}   H_{gi} P_{gi}^{\max} /P_{sys}}  \ge H_{\min} \label{model_1_k}\\
    & u_{gi,t} \in \{ 0,1\} \quad \forall{i}\in\set{I}. \label{model_1_l}
\end{align}%
\label{model_1}%
\end{subequations}%
\allowdisplaybreaks[0]%
Objective~\cref{model_1_a} minimizes the expected operating cost given the cost functions and power outputs of all resources, where $c_{gi}(\cdot)$, $c_{di}(\cdot)$, and $c_{ci}(\cdot)$ are respectively the cost function of G$_i$, ES$_i$ discharging and ES$_i$ charging. 
Binary variable $u_{gi,t}$ in \cref{model_1_l} defines the commitment status of generator at node $i$ and time $t$, i.e., $u_{gi,t}=1$ means the generator is on while $u_{gi,t}=0$ means the generator is off. If a generator is off, i.e., not committed and not synchronized with the grid, it cannot provide power, reserve or inertia. Chance constraints \cref{model_1_b,model_1_c} ensure that the uncertain real-time generator outputs $\bm{P}_{gi,t}$ are contained within their lower and upper limits ($P_{gi}^{\min}$ and $P_{gi}^{\max}$) with a probability of at least $1-\epsilon_{gi}$. Risk level $\epsilon_{gi}>0$ is chosen as a small number and captures the tolerance to constraint violations at this generator. As a result, chance constraints \cref{model_1_b,model_1_c} limit the expected generator output $P_{gi,t}$ and the uncertain real-time balancing contribution $\alpha_{gi,t}\Omega_{p,t}$ simultaneously and thus capture the trade-off between energy and reserve provision.
We opt to not include an additional system-wide constraint on reserve sufficiency as shown in \cite{lubin2019chance} to remain consistent with previous works on chance-constrained market clearing.
Eqs.~\cref{model_1_d}-\cref{model_1_g} model ES operations, i.e., \cref{model_1_d}-\cref{model_1_e} limit power outputs and \cref{model_1_f}-\cref{model_1_g} are energy constraints.
Eq.~\cref{model_1_f} ensures that energy level $e_{it}$ of the ES at node $i$ and time $t$ remains within its technical limits ($E_i^{\min}$ and $ E_i^{\max}$). Eq.~\cref{model_1_g} updates $e_{it}$ based on the charging or discharging power of the ES at node $i$ and time $t$, where $k_i\in[0,1]$ captures the charging and discharging efficiency.
Eq.~\cref{model_1_h} constrains participation factors ${\alpha_{gi,t}} \in [0,1]$, if the generator at node $i$ is committed at time $t$, i.e., ${u_{gi,t}}=1$. The system-wide power balance is enforced in \cref{model_1_i}, ensuring that the sum of generation, ES charging and discharging, and the forecast wind production is equal to system demand $d_{i,t}$ at all times.
Eq. \cref{model_1_j} enforces that the total balancing participation of generators ($\sum\nolimits_{i \in \set{I}} \alpha_{gi,t}\bm{\Omega}_{pt}$) is sufficient to compensate the total wind power forecast error ($\bm{\Omega}_{pt}$). Eq.~\cref{model_1_k} ensures that the inertia provided by committed synchronous generators meets inertia requirement $H_{\min}$, where $u_{gi,t} H_{gi}$ is the inertia provided by the generator at node $i$ and time $t$.
Note that only committed generators can provide inertia to the system. Greek letters in parentheses in \cref{model_1_b}-\cref{model_1_k} denote dual multipliers of the respective constraints.

\subsection{Proposed Model}
\label{subsec:Proposed_Model}

We extend \cref{model_1} to enable ES and RES to provide reserve and/or inertia services as detailed in Sections~\ref{subsec:Uncertainty_Model} and~\ref{subsec:Equivalent_Inertia}, which renders their real-time output uncertain: 
\allowdisplaybreaks
\begin{subequations}
\begin{align}
&\min_{\substack{ \{P_{gi,t}, \alpha _{gi,t} P_{di,t}, P_{ci,t}, \alpha _{di,t}, \\ \alpha _{ci,t}, u_{gi,t}, H_{ei,t} \}_{t \in \set{T}, i \in \set{I}} }}
    \sum\limits_{t \in \set{T}} \sum\limits_{i \in \set{I}} \Big\{ \mathbb{E}_{\bm{\omega}_{pi,t}} [c_{gi} (\bm{P}_{gi,t}, u_{gi,t})  \nonumber \\
    & \hspace{3.5cm} + c_{di}(\bm{P}_{di,t}) + c_{ci}(\bm{P}_{ci,t})]\Big\}   \label{model_2_a}\\
    &\text{s.t. } \forall t\in\set{T}: \cref{model_1_b}-\cref{model_1_c}, \cref{model_1_h}-\cref{model_1_i}, \cref{model_1_l}\nonumber \\
    & (\xi_{i,t}^{+}): \mathbb{P}_{\bm{\omega}_{pi,t}} \left[\bm{P}_{di,t}+ 2H_{ei,t} {f_{\max}^{\prime}} P_{ei}^{\max}/f_0 \le P_{ei}^{\max}  \right]  \nonumber\\
    & \hspace{1cm} \ge 1 - \epsilon_{di} \ \ \forall{i}\in\set{I} \label{model_2_d}\\
    & (\nu_{i,t}^{+}): \mathbb{P}_{\bm{\omega}_{pi,t}} \left[\bm{P}_{ci,t}+ 2H_{ei,t}  {f_{\max }^{\prime}} P_{ei}^{\max}/f_0 \le P_{ci}^{\max}  \right]  \nonumber\\
    & \hspace{1cm} \ge 1 - \epsilon_{ci} \ \ \forall{i}\in\set{I} \label{model_2_e}\\
    & (\eta_{i,t}):e_{i,t} =  e_{i,t-1} + \mathbb{E}_{\bm{\omega}_{pi,t}} \left[\bm{P}_{ci,t} k_i -\bm{P}_{di,t}/k_i \right] \ \forall{i}\in\set{I} \label{model_2_j}\\
    & (\beta_{i,t}^{+}):e_{i,t} \le E_i^{\max} - 2H_{ei,t} \Delta f_{\max} P_{ei}^{\max}/f_0 \ \ \forall{i}\in\set{I} \label{model_2_h}\\
    & (\beta_{i,t}^{-} ):{e_{i,t}} \ge E_i^{\min } + 2H_{ei,t} \Delta f_{\max} P_{ei}^{\max}/f_0 \ \ \forall{i}\in\set{I} \label{model_2_i}\\
    & (\xi_{i,t}^{-}): -P_{di,t} \le 0 \quad \forall{i}\in\set{I} \label{model_2_f}\\
    & (\nu_{i,t}^{-}): -P_{ci,t} \le 0 \quad \forall{i}\in\set{I} \label{model_2_g}\\
    & (\varepsilon_{i,t}):H_{ei,t} \le H_{ei}^{\max} \quad \forall{i}\in\set{I} \label{model_2_k} \\
    & (\rho_{di,t}^{-} ,\rho_{di,t}^{+} ):0 \le \alpha _{di,t} \le 1 \quad \forall{i}\in\set{I} \label{model_2_m} \\
    & (\rho_{ci,t}^{-} ,\rho_{ci,t}^{+} ):0 \le \alpha_{ci,t} \le 1 \quad \forall{i}\in\set{I} \label{model_2_n} \\
    & (\gamma_t):\sum\nolimits_{i \in \set{I}} {\left( {\alpha_{gi,t} + \alpha_{di,t} - \alpha_{ci,t}} \right)}  = 1  \label{model_2_p}\\
    & (\chi_t):\mathbb{P}_{\bm{\omega}_{hi,t}} \!\Big\{ \sum\limits_{i \in \set{I}} {\left( {u_{gi,t} H_{gi} P_{gi}^{\max} {\rm{+}} H_{ei,t} P_{ei}^{\max}{\rm{+}} \bm{H}_{wi,t} P_{wi}^{\max}} \right)  }  \nonumber\\
    & \hspace{0.8cm}  \ge {P_{sys} H_{\min}} \Big\}  \ge 1 - {\epsilon_{hi}}. \label{model_2_q}
\end{align}%
\label{model_2}%
\end{subequations}%
\allowdisplaybreaks[0]%
Compared with \cref{model_1_a}, objective \cref{model_2_a} minimizes the expected operating cost based on the same cost functions, but considers uncertain ES charging and discharging ($\bm{P}_{ci,t}$ and $\bm{P}_{di,t}$). Eqs.~\cref{model_2_d}-\cref{model_2_e} and \cref{model_2_j}-\cref{model_2_i} are the modified ES power and energy constraints. 
Note that we use chance constraints in \cref{model_2_d}-\cref{model_2_e} because the ES charging power limits are selected conservatively and do not reflect actual physical ES limits on charging and, therefore can be violated for a short period. In contrast, energy constraints in \cref{model_2_h}-\cref{model_2_i} reflect the physical storage capacity of ES, which cannot be violated \cite{bienstock2017robust}, and are, therefore, deterministic with a sufficient safety margins given by  $E_i^{\min}$ and $E_i^{\max}$. 
Note that in \cref{model_2_j} the energy level of ES is calculated based on the expected charging and discharging power of ES, which represents the most likely scenario in real-time scheduling considering the normal distribution of forecast error.
Eq.~\cref{model_2_m}-\cref{model_2_n} constrains ${\alpha_{di,t}}$ and ${\alpha_{ci,t}}$, while Eq.~\cref{model_2_p} ensures the reserve sufficiency. Additionally, ES power and energy constraints in \cref{model_2} account for the provision of virtual inertia. In \cref{model_2_d}, \cref{model_2_e}, \cref{model_2_h}, and \cref{model_2_i}, $H_{ei,t}$ is the inertia constant of the ES at node $i$ and time $t$, $\Delta {f_{\max}}$ is the maximum admissible frequency deviation at nadir, while ${f_{\max}^{'}}$ and $f_{0}$ are introduced in Section~\ref{subsec:Inertia_Requirements}.
Thus, $2{H_{ei,t}}  {f_{\max}^{'}} {P_{ei}^{\max}}/{f_0}$ and $2{H_{ei,t}} \Delta {f_{\max}} {P_{ei}^{\max}}/{f_0}$ are, respectively, the inertial power and energy response to the worst-case power imbalance defined by the system operator \cite{paturet2020economic,markovic2018lqr}. 
Eq.~\cref{model_2_k} limits the inertia constant of ES to $H_{ei}^{\max}$, and \cref{model_2_q} ensures the sufficiency of inertia provided by generators, ES and RES.  Model \cref{model_1} and \cref{model_2} can be modified to capture additional constraints, e.g., ramping constraints, which we have omitted due to their minimal impact on the price analysis below.

\subsection{Deterministic Equivalent of the CC-UC Model}
\label{subsec:Deterministic_Equivalent}

The CC-UC model in \cref{model_2} contains expectation ($\mathbb{E}$) and probability ($\mathbb{P}$) operators, which are reformulated into computationally tractable forms as follows.

\subsubsection{Expected Generation Cost}
The production costs of each controllable generator $c_{gi}(\cdot)$ and ES discharging and charging costs $c_{di}(\cdot)$, $c_{ci}(\cdot)$ are given as:
\begin{align}
    c_{gi} (\bm{P}_{gi,t}, u_{gi,t}) &= u_{gi,t} c_{0i} +c_{1i}\bm{P}_{gi,t} + c_{2i} (\bm{P}_{gi,t})^2   \label{cgi}\\
    c_{di} (\bm{P}_{di,t}) &= c_{di} \bm{P}_{di,t}  \label{cdi}\\ 
    c_{ci} (\bm{P}_{ci,t}) &= c_{ci} \bm{P}_{ci,t},  \label{cci}
\end{align}
where $c_{0i}$, $c_{1i}$, $c_{2i}$, $c_{ci}$ and $c_{di}$ are cost coefficients. Note that instead of modeling ES cost in \cref{cdi,cci} with a more complex cost function which captures charging and discharging cycles and degradation effects \cite{xu2016modeling}, we approximate this relationship via linear cost factor  $c_{ci}$ and $c_{di}$. 
Thus, using \cref{cgi,cdi,cci} and recalling that $\mathbb{E}[\bm{x}^2]=\mathbb{E}[\bm{x}]^2 + \Var[\bm{x}]$ for any random $\bm{x}$, 
the expected cost of individual generator and ES at time $t$ are given as:
\begin{align}
    C_{gi,t} & = \mathbb{E}_{\bm{\omega}_{pi,t}} \left[c_{gi} \left( {\bm{P}_{gi,t}, u_{gi,t} } \right) \right] \nonumber \\
    & =c_{0i} u_{gi,t}  +c_{1i}\left( {P_{gi,t}+\mathrm{M}_{pt} \alpha _{gi,t}} \right) \nonumber \\  
    & + c_{2i} \left[ (P_{gi,t} + \mathrm{M}_{pt} \alpha _{gi,t})^2 + \Sigma_{pt}^2 \alpha _{gi,t}^2 \right]  \label{Cgi} \\
    C_{ei,t }& = \mathbb{E}_{\bm{\omega}_{pi,t}} \left[c_{di} \left(\bm{P}_{di,t} \right) +c_{ci} \left(\bm{P}_{ci,t} \right) \right] \nonumber \\
    & = {c_{di}\left(P_{di,t} + \mathrm{M}_{pt} \alpha _{di,t} \right) + c_{ci}\left( {P_{ci,t} + \mathrm{M}_{pt} \alpha _{ci,t}} \right)}.   \label{Cei}
\end{align}
Therefore, the expected cost of all generators is $C_G = \sum\nolimits_{t \in \set{T}} \sum\nolimits_{i \in \set{I}}  C_{gi,t}$, and the expected cost of all ES is $C_{ES} =  \sum\nolimits_{t \in \set{T}} \sum\nolimits_{i \in \set{I}} C_{ei,t}$.

\subsubsection{Chance Constraints}

To deal with chance constraints in \cref{model_2}, we introduce the following notations: $\hat \delta_{gi} = \Phi^{-1} (1 - \epsilon_{gi}) \Sigma_{pt} - \mathrm{M}_{pt}$, $\hat \delta_{di} = \Phi^{-1} (1 - \epsilon_{di}) \Sigma_{pt} - \mathrm{M}_{pt}$, $\hat \delta_{ci} = \Phi^{-1} (1 - \epsilon_{ci}) \Sigma_{pt} - \mathrm{M}_{pt}$, $\hat \delta_{hi} = \Phi ^{-1} (1 - \epsilon_{hi}) \Sigma_{ht} -\mathrm{M}_{ht}$, where $\Phi^{-1}(\cdot)$ is the inverse cumulative distribution function of the standard normal distribution.
Following the reformulation presented in, e.g., \cite{bienstock2014chance,kuang2018pricing,dvorkin2019chance}, we use $\hat \delta_{gi}$, $\hat \delta_{di}$, $\hat \delta _{ci}$ and $\hat \delta_{hi}$ to reformulate 
\cref{model_1_b}-\cref{model_1_c}, \cref{model_2_d}-\cref{model_2_e} and \cref{model_2_q} into \cref{model_determined_b}-\cref{model_determined_g}.
Furthermore, if $\bm{\omega}_{pi,t}$ or $\bm{\omega}_{hi,t}$ follows other distributions, e.g., Student’s t-distribution, the aforementioned auxiliary variables will be formed differently (see \cite{roald2015security}), but the reformulations of the chance constraints in \cref{model_2} will remain the same form as in \cref{model_determined}.

\subsubsection{Deterministic CC-UC Equivalent}

The reformulated objective function and chance constraints
lead to the following deterministic equivalent of \cref{model_2}:
\allowdisplaybreaks
\begin{subequations}
\begin{align}
    &\min_{\substack{\{P_{gi,t}, P_{di,t}, P_{ci,t},u_{gi,t},\alpha _{gi,t}\\  \alpha _{di,t},\alpha _{ci,t},H_{ei,t}\}_{t \in \set{T}, i \in \set{I}} }}\  C_G+C_{ES} \label{model_determined_a}\\
    &\text{s.t. }\forall t\in\set{T}, \forall i\in\set{I}: 
    \text{\cref{model_1_h}-\cref{model_1_i}, \cref{model_1_l}, \cref{model_2_h}-\cref{model_2_p}}\nonumber \\
    & (\mu_{gi,t}^ {+} ): P_{gi,t} \le u_{gi,t} P_{gi}^{\max } - {{\hat \delta }_{gi}} \alpha _{gi,t} \label{model_determined_b}\\
    & (\mu_{gi,t}^ {-} ): - P_{gi,t} \le  - u_{gi,t} P_{gi}^{\min} - {{\hat \delta}_{gi}} \alpha _{gi,t}   \label{model_determined_c}\\
    & (\xi_{i,t}^ {+} ): P_{di,t} + 2H_{ei,t} {f_{\max }^{'}} P_{ei}^{\max}/f_0 \le P_{ei}^{\max} {\rm{-}} {{\hat \delta}_{di}} \alpha _{di,t}  \label{model_determined_d}\\
    & (\nu_{i,t}^ {+} ): P_{ci,t} + 2H_{ei,t} {f_{\max }^{'}} P_{ei}^{\max}/f_0 \le P_{ci}^{\max} {\rm{-}} {{\hat \delta}_{ci}} \alpha _{ci,t}  \label{model_determined_e}\\
    & (\eta_{i,t}): e_{i,t} = e_{i,t-1} + (P_{ci,t} + \mathrm{M}_{pt} \alpha_{ci,t} )  {k_i}  \nonumber \\
    & \hspace{1.6cm} - (P_{di,t} + \mathrm{M}_{pt} \alpha_{di,t}) /{k_i}\label{model_determined_f}\\
    & (\chi_t): \sum\limits_{i \in \set{I}} {\left( {u_{gi,t} H_{gi} P_{gi}^{\max} {\rm{+}} H_{ei,t} P_{ei}^{\max} {\rm{+}} (H_{wi,t} {\rm{+}} {\hat \delta}_{hi})P_{wi}^{\max}} \right)}  \nonumber \\
    & \hspace{0.8cm} \ge P_{sys} H_{\min}. \label{model_determined_g}
\end{align}%
\label{model_determined}%
\end{subequations}%
\allowdisplaybreaks[0]%
Note that ${{\hat \delta }_{gi}} \alpha _{gi,t}$ effectively determines the amount of balancing reserve that the generator at node $i$ procures.
Since all the constraints in \cref{model_determined} are linear, the solving complicity of \cref{model_determined} is comparable to traditional UC models, which can be efficiently solved by the system operators.

\subsubsection{Convex CC-UC Equivalent}

Model \cref{model_determined} is a mixed-integer quadratic program (MIQP) due to the presence of binary variables $u_{gi,t}$ and its quadratic objective function. This model is non-convex, but can be solved by modern off-the-shelf solvers (e.g., CPLEX, Gurobi). To obtain dual variables needed to compute marginal prices, we need to convert \cref{model_determined} into an equivalent convex quadratic program (QP):
\allowdisplaybreaks
\begin{subequations}
\begin{align}
&\min_{\substack{\{P_{gi,t}, P_{di,t}, P_{ci,t},u_{gi,t},\alpha _{gi,t}\\  \alpha _{di,t},\alpha _{ci,t},H_{ei,t}\}_{t \in \set{T}, i \in \set{I}} }}\  C_G+C_{ES} \label{model_QP_a}\\
    \text{s.t. }&\forall t\in\set{T}, \forall i\in\set{I}: \text{ \cref{model_1_h}-\cref{model_1_i}, \cref{model_1_l}, \cref{model_2_h}-\cref{model_2_p}},\nonumber\\
    &  \text{\cref{model_determined_b}-\cref{model_determined_g}} \nonumber \\
    & (\kappa _{i,t}): u_{gi,t} = u_{gi,t}^{*},   \label{model_QP_b}
\end{align}%
\label{model_QP}%
\end{subequations}%
\allowdisplaybreaks[0]%
by fixing binary variables $u_{gi,t}=u_{gi,t}^*$, where $u_{gi,t}^*$ are the results obtained after \cref{model_determined} has been solved once by a MIQP solver. This approach follows the results in \cite{kuang2019pricingquadratic}, where the authors showed that the optimal solution of the MIQP is equal to the optimal solution of the QP derived through the aforementioned method. Further, according to \cite{o2005efficient, kuang2019pricingquadratic}, this QP has dual multipliers that have the traditional economic interpretation as prices and clear the market in the presence of nonconvexities.
In this paper, although we assign the same dual multiplier to constraints with the same structure and function for all models (including \cref{model_1}, \cref{model_2}, \cref{model_determined} and \cref{model_QP}) to keep the notations concise, only the dual multipliers of the final convex QP model \cref{model_determined} will be used to derive the prices in the next section.

Alternative approaches to convexify \cref{model_QP} are possible \cite{gribik2007market}. For example, instead of fixing binaries $u_{gi,t}$ they can be relaxed to the unit interval $u_{gi,t}\in[0,1]$. Alternatively, \cref{model_QP} can be approximated with a tight convex hull. 
As the former method introduces inaccuracies and the the latter method is computationally demanding, we omit a detailed discussion of these methods in this paper.

\section{Energy, Reserve and Inertia Prices}
\label{sec:prices}

The convex CC-UC model in \cref{model_QP} is used to obtain and analyze the following three prices. First, the active power price is derived from dual multiplier $\lambda_t$ of the power balance constraint in \cref{model_1_i}. Second, the price of balancing regulation is derived from dual multiplier $\gamma_t$ of the reserve sufficiency requirement in \cref{model_2_p}. Third, the price of inertia is derived from dual multiplier $\chi_t$ of the inertia requirement \cref{model_determined_g}.
In the following, we derive Propositions~\ref{prop:active_power_price}-\ref{prop:inertia_price} to highlight individual price components, which facilitate the analysis of whether and how certain constraints will affect the prices.

\subsection{Energy (Active Power) Pricing}
\label{subsec:price_for_active_power}

\begin{proposition}
\label{prop:active_power_price}
Consider the model in \cref{model_QP}. Let $\lambda_t$ be the active power price defined as dual multipliers of constraint \cref{model_1_i}. Then $\lambda_t$ is given by:
\begin{align}
    \lambda _t {\rm{=}} \frac{{\sum\limits_{i \in \set{I}} {\left[ {(\mu _{i,t}^ {+}  {\rm{-}} \mu _{i,t}^{-}  {\rm{+}} c_{1i})/2c_{2i} {\rm{+}} \mathrm{M}_{pt} \alpha _{gi,t} {\rm{-}} P_{di,t} {\rm{+}} P_{ci,t} {\rm{+}} d_{i,t}} \right]} }}{{\sum\nolimits_{i \in \set{I}} {1/(2c_{2i})} }},
    \label{lambda_t}
\end{align}
where $\mu _{i,t}^ {+}$ and $\mu _{i,t}^ {-}$ are the dual multipliers of \cref{model_1_b} and \cref{model_1_c}.
\end{proposition}

\begin{proof}
The Karush-Kuhn-Tucker (KKT) optimality conditions for the model in \cref{model_QP} are shown in \cref{KKT}:
\allowdisplaybreaks
\begin{subequations}
\begin{align}
    (P_{gi,t})\!:\ &c_{1i} {\rm{+}} 2c_{2i} \left( {P_{gi,t} {\rm{+}} \mathrm{M}_{pt} \alpha _{gi,t}} \right){\rm{+}} \mu _{i,t}^ {+}  {\rm{-}} \mu_{gi,t}^ {-} {\rm{-}} \lambda_t {\rm{=}} 0 \label{KKT_a}\\
    (P_{di,t})\!:\ &c_{di} + \xi _{i,t}^ {+}  - \xi _{i,t}^{-} - \eta_{i,t}  /k_i - \lambda_t = 0 \label{KKT_b}\\
    (P_{ci,t})\!:\ & - c_{ci} + \nu _{i,t}^{+} - \nu _{i,t}^{-} + \eta_{i,t}   k_i + \lambda_t = 0 \label{KKT_c}\\
    (e_{i,t})\!:\ &\beta_{i,t}^{+} - \beta_{i,t}^{-} + \eta_{i,t} - \eta_{i,t + 1} = 0 \label{KKT_d}\\
    (\alpha _{gi,t})\!:\ &c_{1i} \mathrm{M}_{pt} +2c_{2i} \left[ {\mathrm{M}_{pt} P_{gi,t}+\alpha _{gi,t} \left( {\Sigma_{pt}^2+ \mathrm{M}_{pt}^2} \right)} \right] \nonumber \\
    & + \mu _{i,t}^{+} {\hat \delta }_{gi} + \tau_{gi,t}^{-} {\hat \delta }_{gi} + \rho_{gi,t}^{+} - \rho_{gi,t}^{-} - \gamma_t = 0 \label{KKT_e} \\
    (\alpha _{di,t})\!:\ &c_{di} \mathrm{M}_{pt} + \xi_{i,t}^{+} {\hat \delta }_{di} + \rho_{di,t}^{+} -\rho_{di,t}^{-} - \gamma_t = 0 \label{KKT_f}\\
    (\alpha_{ci,t})\!:\ &c_{ci} \mathrm{M}_{pt} +\nu_{i,t}^{+}{\hat \delta }_{ci}+\rho _{ci,t}^{+} - \rho_{ci,t}^{-} + \gamma_t=0 \label{KKT_g}\\
    (u_{gi,t})\!:\ &{c_{0i}} {\rm{-}} \mu _{i,t}^{+}  {\rm{+}} \mu _{i,t}^{-}  {\rm{+}} \kappa_{i,t} {\rm{-}} \rho _{gi,t}^{+} {\rm{-}} \chi_t H_{gi} P_{gi}^{\max} = 0 \label{KKT_h}\\
    (H_{ei,t})\!:\ & \varepsilon_{i,t}- \chi_t P_{ei}^{\max}+2\left( {\xi_{i,t}^{+}  + \nu_{i,t}^{+} } \right){f_{\max }^{\prime}}P_{ei}^{\max}/f_0 \nonumber \\
    & + 2\left( {\beta _{i,t}^{+} + \beta _{i,t}^{-}} \right) \Delta f_{\max } P_{ei}^{\max}/f_0 = 0 \label{KKT_i}
\end{align}%
\label{KKT}%
\end{subequations}%
\allowdisplaybreaks[0]%
From \cref{KKT_a} we obtain: 
\begin{align}
    P_{gi,t} = (- \mu _{i,t}^{+} + \mu _{i,t}^{-} + \lambda _t - c_{1i}) /2c_{2i} - \mathrm{M}_{pt} \alpha _{gi,t} \label{lambda_tp}.
\end{align} 
Substituting \cref{lambda_tp} into \cref{model_1_i} returns $\lambda_t$ as in \cref{lambda_t}.
\end{proof}

Notably, $\lambda_t$ can also be expressed directly from \cref{KKT_a}, \cref{KKT_b} and \cref{KKT_c} as:
\begin{align}
    \lambda _t &= c_{1i} + 2c_{2i}(P_{gi,t} + \mathrm{M}_{pt} \alpha _{gi,t}) + \mu _{i,t}^{+}  - \mu _{i,t}^{-} \label{lambda_g}\\
    \lambda _t &= c_{di} + \xi _{i,t}^{+}  - \xi _{i,t}^{-}  - \eta _{i,t} {\rm{/}} k_i \label{lambda_d}\\
    \lambda _t &= c_{ci} - \nu _{i,t}^{\rm{ + }} + \nu _{i,t}^{-} - \eta _{i,t}   k_i. \label{lambda_c}
\end{align}
The difference between \cref{lambda_t} and \cref{lambda_g}-\cref{lambda_c} is that \cref{lambda_t} is derived from the perspective of the whole system, while \cref{lambda_g}-\cref{lambda_c} are derived from the resource perspective and capture the marginal cost of generators and ES discharging/charging. In an ideal market, \cref{lambda_g}-\cref{lambda_c} and \cref{lambda_t} are such that marginal cost is equal to the system marginal price $\lambda_t$ \cite{Kirschen2018}.

\subsection{Reserve Pricing}
\label{subsec:price_for_reserve}

\begin{proposition}
\label{prop:reserve_price}
Consider the model in \cref{model_QP}. Let $\gamma_t$ be the reserve price defined as dual multipliers of constraint \cref{model_2_p}. Then $\gamma_t$ is given by:
\begin{align}
    \gamma _t &{\rm{=}} \frac{{\sum\limits_{i \in \set{I}} { [ {c_{1i}\mathrm{M}_{pt} {\rm{+}} ( {\mu _{i,t}^{+}  {\rm{+}} \mu _{i,t}^{-} } ){\hat \delta }_{gi} {\rm{+}} \rho_{gi,t}^{+} {\rm{-}} \rho_{gi,t}^{-}} ]/2c_{2i}}{\rm{+}} \mathrm{M}_{pt} P_{gi,t} }}{{\sum\nolimits_{i \in \set{I}} {1/(2c_{2i})} }} \nonumber\\
    & + \frac{{ \left [1 - \sum\nolimits_{i \in \set{I}} {( {{\alpha _{di,t}} - {\alpha_{ci,t}}} ) }\right] \left( {\Sigma_{pt} ^2 + \mathrm{M}_{pt} ^2} \right) }}{{\sum\nolimits_{i \in \set{I}} {1/\left( {2c_{2i}} \right)} }}, \label{gamma_t}
\end{align}
where $\rho_{gi,t}^{+}$ and $\rho_{gi,t}^{-}$ are the dual multipliers of \cref{model_1_h}.
\end{proposition}

\begin{proof}
From \cref{KKT_e} we obtain:
\begin{align}
    \alpha_{gi,t} &= \frac{{-c_{1i} \mathrm{M}_{pt} - 2c_{2i} \mathrm{M}_{pt} P_{gi,t}- (\tau_{gi,t}^{+} + \tau_{gi,t}^{-} ){{\hat \delta}_{gi}} }}{{2c_{2i}(\Sigma_{pt} ^2+\mathrm{M}_{pt}^2)}} \nonumber \\
    &+ \frac{{\gamma_t -{\rho_{gi,t}^{+}} + {\rho_{gi,t}^{-}}}}{{2c_{2i}(\Sigma_{pt} ^2+\mathrm{M}_{pt} ^2)}}. 
    \label{gamma_t_alpha}
\end{align} 
Substituting \cref{gamma_t_alpha} into \cref{model_2_p} returns $\gamma_t$ as in \cref{gamma_t}.
\end{proof}

Also, $\gamma_t$ can be expressed from \cref{KKT_e}, \cref{KKT_f} and \cref{KKT_g}:
\begin{align}
    \gamma _t &= c_{1i} \mathrm{M}_{pt} {\rm{+}} 2c_{2i}\left [\mu _p P_{gi,t} {\rm{+}} \alpha_{gi,t}(\sigma _p^2 {\rm{+}} \mathrm{M}_{pt}^2) \right] \nonumber \\
    &\hspace{0.4cm} + \mu _{i,t}^{+} {{\hat \delta }_{gi}} + \mu _{i,t}^{-} {{\hat \delta }_{gi}} + \rho_{gi,t}^{+}  - \rho_{gi,t}^{-} \label{gamma_g}\\
    \gamma _t &= c_{di} \mathrm{M}_{pt} + \xi _{i,t}^{+} {{\hat \delta}_{di}} + \rho _{di,t}^{+}  - \rho _{di,t}^{-} \label{gamma_d} \\
    \gamma_t &= -c_{ci} \mathrm{M}_{pt} - \nu _{i,t}^{\rm{+}}{{\hat \delta}_{ci}} - \rho _{ci,t}^{+}  + \rho _{ci,t}^{-}. \label{gamma_c}
\end{align}
Similarly to the energy price, \cref{gamma_t} accounts for the system perspective, while \cref{gamma_g}-\cref{gamma_c} represent the resource perspective. 
For example, if the generator at node $i$ is the marginal reserve provider, i.e., $\alpha_{gi,t}>0$, with non-binding power constraints in \cref{model_determined_b} and \cref{model_determined_c} and reserve constraints in  \cref{model_1_h}, then dual multipliers $\tau_{gi,t}^{+}$, $\tau_{gi,t}^{-}$, $\rho_{gi,t}^{+}$ and $\rho_{gi,t}^{-}$ are equal to zero. As a result, \cref{gamma_g} becomes:
\begin{align}
    \gamma_t = c_{1i}\mathrm{M}_{pt} + 2c_{2i}  \left[\mathrm{M}_{pt} P_{gi,t} + \alpha_{gi,t} (\Sigma_{pt}^2+ \mathrm{M}_{pt}^2) \right], \label{gamma_simple}
\end{align}
i.e., the reserve price is determined by the cost coefficients and power output of the generator at node $i$ and system-wide uncertainty parameters $\Sigma_{pt}$ and $\mathrm{M}_{pt}$. Comparing \cref{gamma_simple} with \cref{lambda_g} also reveals that both $\lambda_t$ and $\gamma_t$ directly depend on the systematic wind power forecast error (captured by mean $\mathrm{M}_{pt}$), while $\gamma_t$ additionally depends on standard deviation $\Sigma_{pt}$ explicitly.

Further, if ES at node $i$ is a marginal reserve provider at time $t$, then \cref{gamma_d} results in  $\gamma_t = c_{di} \mathrm{M}_{pt}$, i.e., the reserve price is equal to the marginal discharging cost of the ES at node $i$ multiplied by the mean value of wind power forecast error and does not directly depend on discharging power $P_{di,t}$ and participation factor $\alpha_{di,t}$. However, as per \cref{gamma_t}, $\alpha_{di,t}$ and $\alpha_{ci,t}$ will also affect $\gamma_t$, because the change in $\alpha_{di,t}$ and $\alpha_{ci,t}$ may change the value of $\alpha_{gi,t}$.

\subsection{Inertia Pricing}
\label{subsec:price_for_inertia}

\begin{proposition}
\label{prop:inertia_price}
Consider the model in \cref{model_QP}. Let $\chi_t$ be the reserve price defined as dual multipliers of constraint \cref{model_determined_g}. If \cref{model_determined_g} is binding, then the non-zero $\chi_t$ is given by:
\begin{align}
    \chi_t {\rm{=}} \frac{{\sum\nolimits_{i \in \set{I}} {\left[{u_{gi,t}}({c_{0i}} {\rm{-}} \mu_{i,t}^+  {\rm{+}} \mu_{i,t}^- {\rm{+}} \kappa_{i,t} {\rm{-}} \rho_{i,t}^+)\right]} }}{{P_{sys} H_{\min} {\rm{-}} \sum\limits_{i \in \set{I}} {[ {H_{ei,t}P_{ei}^{\max} {\rm{+}} (H_{wi,t} {\rm{+}} {{\hat \delta}_{wi}})P_{wi}^{\max}} ]}}}. \label{chi_t}
\end{align}
\end{proposition}

\begin{proof}
From \cref{KKT_h} we obtain:
\begin{align}
H_{gi} = (c_{0i} - \mu_{i,t}^+  + \mu_{i,t}^- + \kappa _{i,t} - \rho_{gi,t}^+ )/(\chi_t P_{gi}^{\max}).\label{H_gi}
\end{align}
Whenever $\chi_t\neq0$, constraint \cref{model_determined_g} is binding. 
Thus, by substituting \cref{H_gi} into \cref{model_determined_g} we obtain \cref{chi_t}.
\end{proof}
Also, $\chi_t$ can be expressed from the perspective of generator and ES respectively based on \cref{KKT_h} and \cref{KKT_i}:
\begin{align}
    \chi_t &{\rm{=}} {\left( {c_{0i} {\rm{-}} \tau_{gi,t}^{+}  {\rm{+}} \tau_{gi,t}^{-} {\rm{+}} \kappa_{i,t} {\rm{-}} \rho_{gi,t}^{+} } \right)}/({H_{gi} P_{gi}^{\max})} \label{chi_g}\\
    \chi _t &{\rm{=}} \frac {2}{f_0}\left[ {(\xi_{i,t}^{+}  {\rm{+}} \nu_{i,t}^{+} ){f_{\max}^{\prime}} {\rm{+}} (\beta_{i,t}^{+}  {\rm{+}}\beta_{i,t}^{-})\Delta {f_{\max}}} \right] {\rm{+}} \frac{\varepsilon_{i,t}}{P_{ei}^{\max}}. \label{chi_e}
\end{align}

If constraints \cref{model_1_h}, \cref{model_determined_b} and \cref{model_determined_c} are non-binding, then the dual parameters of these constraints ($\mu _{i,t}^{+}$, $\tau_{gi,t}^{-}$ and $\rho _{gi,t}^{+}$) are zero, reducing \cref{chi_g}  to:
\begin{align}
    \chi_t = (c_{0i} + \kappa_{i,t}) / (H_{gi}   P_{gi}^{\max}), \label{chi_simple}
\end{align}
i.e., the inertia price depends on no-load cost $c_{0i}$ and commitment price $\kappa_{i,t}$ of the generator at node $i$ and time $t$ and does not depend on its output power $P_{gi,t}$ or its reserve participation $\alpha_{gi,t}$. Similarly, if ES constraints at node $i$ (\cref{model_2_h}, \cref{model_2_i}, \cref{model_determined_d} and \cref{model_determined_e}) are non-binding, then all the dual variables ($\xi_{i,t}^{+}$, $\nu_{i,t}^{+}$, $\beta_{i,t}^{+}$ and $\beta_{i,t}^{-}$) are zero, which leads to $\chi_t=0$. This result demonstrates that the ES opportunity cost from providing virtual inertia is zero when ES constraints are non-binding. Thus, there is no additional cost for the ES to provide virtual inertia.

\subsection{Revenue, Cost and Profit of Producers}
\label{subsec:Revenue}

Given the energy, reserve and inertia prices ($\lambda_t,\gamma_t,\chi_t$) in  Propositions~\ref{prop:active_power_price}-\ref{prop:inertia_price}, the expected revenue of generators, ES and RES are:
\begin{align}
    R_{gi}(\lambda_t,\gamma_t,\chi_t) & {\rm{=}} \sum\nolimits_{t \in \set{T}} {\left( {\lambda _t  P_{gi,t} {\rm{+}} \gamma_t  \alpha_{gi,t} {\rm{+}} \chi_t  u_{gi,t} H_{gi} I_{gi}} \right)} \label{Rg}\\
    R_{ei}(\lambda_t,\gamma_t,\chi_t) & {\rm{=}} \sum\nolimits_{t \in \set{T}} {\left( {\lambda _t  P_{di,t} {\rm{+}}\gamma_t \alpha_{di,t} {\rm{+}} \chi_t  H_{ei,t} I_{ei}} \right)}  \label{Re}\\
    R_{wi}(\lambda_t,\gamma_t,\chi_t) & {\rm{=}} \sum\nolimits_{t \in \set{T}} {\chi_t  H_{wi,t} I_{wi}},  \label{Rw}
\end{align}
where $I_{gi} {\rm{=}} P_{gi}^{\max} / P_{sys}$, $I_{ei} {\rm{=}} P_{ei}^{\max} / P_{sys}$, $I_{wi} {\rm{=}} P_{wi}^{\max} / P_{sys}$. Note that the revenue of providing inertia is scaled by the relative inertia contribution $I_{gi}$, $I_{ei}$ and $I_{wi}$ of each generator, ES or RES. On the other hand, the expected cost of individual generator and ES as in \cref{Cgi,Cei} are given as
$C_{gi} = \sum\nolimits_{t \in \set{T}} C_{gi,t}$ and $C_{ei} = \sum\nolimits_{t \in \set{T}} C_{ei,t}$.

Assuming that the marginal RES production is zero, the expected profit of individual conventional ($\Pi_{gi}$), ES ($\Pi_{ei}$) and RES ($\Pi_{wi}$) resource is:
\begin{align}
    \Pi_{gi}(\lambda_t,\gamma_t,\chi_t) &= R_{gi}(\lambda_t,\gamma_t,\chi_t)-C_{gi} \label{PRg}\\
    \Pi_{ei}(\lambda_t,\gamma_t,\chi_t) &= R_{ei}(\lambda_t,\gamma_t,\chi_t)-C_{ei} \label{PRe}\\
    \Pi_{wi}(\lambda_t,\gamma_t,\chi_t) &= R_{wi}(\lambda_t,\gamma_t,\chi_t).  \label{PRw}
\end{align}

\subsection{Competitive Equilibrium}
\label{subsec:Equilibrium}

Next, we prove that the prices derived in  Propositions~\ref{prop:active_power_price}-\ref{prop:inertia_price} are efficient and incentive compatible, i.e., they constitute a competitive equilibrium. Similarly to  \cite{dvorkin2019chance}, we define the competitive equilibrium as follows:

\begin{definition}
\label{Definiation:Equilibrium}
A competitive equilibrium for the stochastic market defined by \cref{model_QP} is a set of prices \{$\lambda_t$, $\gamma_t$, $\chi_t$, $\forall t\in\set{T}$\} and a set of dispatch decisions \{$P_{gi,t}$, $P_{di,t}$, $P_{ci,t}$, $\alpha_{gi,t}$, $\alpha_{di,t}$, $\alpha_{ci,t}$, $u_{gi,t}$, $H_{ei,t}$, $\forall i\in\set{I}, \forall t\in\set{T}$\} that satisfy two conditions. First, the market clears such that power production and demand are balanced and the reserve and inertia requirements are met. Second, all producers maximize their profits, so that there is no incentive to deviate from the market outcomes.
\end{definition}

To show that the prices from Propositions~\ref{prop:active_power_price}-\ref{prop:inertia_price} lead to a competitive equilibrium, we model each generator and ES in a risk-neutral, profit-maximizing manner.
Thus, each generator chooses $P_{gi,t}$, $\alpha _{gi,t}$ and $u_{gi,t}$ using the following optimization:
\allowdisplaybreaks
\begin{subequations}
\begin{align}
&\max_{\substack{\{P_{gi,t},u_{gi,t},\alpha _{gi,t}\}_{t \in \set{T}, i \in \set{I}} }}\ \Pi_{gi}(\pi_{gi,t},\varphi_{gt},\psi_{gt}) \label{Competitive_g_a}\\
    & \text{s.t. } \forall t\in\set{T}: \cref{model_1_h}, \cref{model_determined_b}, \cref{model_determined_c}, \cref{model_QP_b},  \nonumber 
\end{align}%
\label{Competitive_g}%
\end{subequations}%
\allowdisplaybreaks[0]%
where $\Pi_{gi}$ denotes the profit function of the generator at node $i$ and {$\pi_{gi,t}$, $\varphi_{gt}$, $\psi_{gt}$} are the energy, reserve and inertia prices at node $i$ and time $t$. Similarly, each ES determines its $P_{di,t}$, $P_{ci,t}$, $\alpha_{di,t}$, $\alpha_{ci,t}$ and $H_{ei,t}$ by solving:
\allowdisplaybreaks
\begin{subequations}
\begin{align}
&\max_{\substack{\{ P_{di,t}, P_{ci,t}, \alpha_{di,t}, \alpha_{ci,t}, H_{ei,t} \}_{t \in \set{T}, i \in \set{I}} }}\  \Pi_{ei}(\pi_{ei,t},\varphi_{et},\psi_{et}) \label{Competitive_e_a}\\
    &\text{s.t. } \forall t\in\set{T}:
    \cref{model_2_h}-\cref{model_2_p},\cref{model_determined_d}-\cref{model_determined_f},\nonumber
\end{align}%
\label{Competitive_e}%
\end{subequations}%
\allowdisplaybreaks[0]%
where $\Pi_{ei}$ denotes the profit function of the ES at node $i$ and {$\pi_{ei,t}$, $\varphi_{et}$, $\psi_{et}$} are the energy, reserve and inertia prices at node $i$ at time $t$. 

\begin{theorem}
\label{Theorem:Equilibrium}
Let \{$P_{gi,t}^*$, $P_{di,t}^*$, $P_{ci,t}^*$, $\alpha_{gi,t}^*$, $\alpha_{di,t}^*$, $\alpha_{ci,t}^*$, $u_{gi,t}^*$, $H_{ei,t}^*$, $\forall i\in\set{I}, \forall t\in\set{T}$\} be an optimal solution of \cref{model_QP} and let \{$\lambda_t^*$, $\gamma_t^*$, $\chi_t^*$, $\forall t\in\set{T}$\} be the dual variables of \cref{model_QP}, then the set of production levels and prices \{$P_{gi,t}^*$, $P_{di,t}^*$, $P_{ci,t}^*$, $\alpha_{gi,t}^*$, $\alpha_{di,t}^*$, $\alpha_{ci,t}^*$, $u_{gi,t}^*$, $H_{ei,t}^*$, $\pi_{gi,t}$, $\varphi_{gt}$, $\psi_{gt}$, $\pi_{ei,t}$, $\varphi_{et}$, $\psi_{et}$, $\forall i\in\set{I}, \forall t\in\set{T}$\}\ constitutes a competitive equilibrium if $\pi_{gi,t}=\pi_{ei,t}=\lambda_t^*$, $\forall i\in\set{I},\ \forall t\in\set{T}$, $\varphi_{gt}=\varphi_{et}=\gamma_t^*$,  $\forall t\in\set{T}$, and $\psi_{gt}=\psi_{et}=\chi_t^*$, $\forall t\in\set{T}$.
\end{theorem}

\begin{proof}
The KKT optimality conditions for each profit-maximizing generator in \cref{Competitive_g} are:
\allowdisplaybreaks
\begin{subequations}
\begin{align}
    (P_{gi,t})\!:\ & c_{1i} {\rm{+}} 2c_{2i} \left( {P_{gi,t} {\rm{+}} \mathrm{M}_{pt} \alpha _{gi,t}} \right){\rm{+}} \mu _{i,t}^{+}  {\rm{-}} \tau_{gi,t}^{-} {\rm{-}} \pi_{gi,t} {\rm{=}} 0 \label{KKT_Competitive_g_a}\\
    (\alpha _{gi,t}) \!:\ & c_{1i} \mathrm{M}_{pt} +2c_{2i} \left[ {\mathrm{M}_{pt} P_{gi,t}+\alpha _{gi,t} \left( {\Sigma_{pt}^2+ \mathrm{M}_{pt}^2} \right)} \right]\nonumber \\
    & + \mu _{i,t}^{+} {\hat \delta }_{gi}  + \tau_{gi,t}^{-} {\hat \delta }_{gi} + \rho_{gi,t}^{+} - \rho_{gi,t}^{-} - \varphi_{gt} = 0 \label{KKT_Competitive_g_b} \\   
    (u_{gi,t}) \!:\ & {c_{0i}} {\rm{-}} \mu _{i,t}^{+}  {\rm{+}} \mu _{i,t}^{-}  {\rm{+}} \kappa_{i,t} {\rm{-}} \rho _{gi,t}^{+} {\rm{-}} \psi_{gt} H_{gi} P_{gi}^{\max}  {\rm{=}} 0. \label{KKT_Competitive_g_c}
\end{align}%
\label{KKT_Competitive_g}%
\end{subequations}%
\allowdisplaybreaks[0]%
Using \cref{KKT_Competitive_g}, we can express $\pi_{gi,t}$, $\varphi_{gt}$ and $\psi_{gt}$ as:
\allowdisplaybreaks
\begin{align}
    \pi_{gi,t} &= c_{1i} + 2c_{2i}(P_{gi,t} + \mathrm{M}_{pt} \alpha_{gi,t}) + \mu _{i,t}^{+}  - \tau_{gi,t}^{-} \label{pi_g}\\
    \varphi_{gt} &= c_{1i} \mathrm{M}_{pt} {\rm{+}} 2c_{2i}\left [\mathrm{M}_{pt} P_{gi,t} {\rm{+}} \alpha_{gi,t}(\sigma _p^2 {\rm{+}} \mathrm{M}_{pt}^2) \right] \nonumber \\
    &\hspace{0.4cm} + \tau_{gi,t}^{+} {{\hat \delta}_{gi}} + \tau_{gi,t}^{-} {{\hat \delta}_{gi}} + \rho_{gi,t}^{+}  - \rho_{gi,t}^{-}  \label{varphi_g}\\
    \psi_{gt} &=  {\left( {c_{0i} - \tau_{gi,t}^{+}  + \tau_{gi,t}^{-}  + \kappa_{i,t} - \rho_{gi,t}^{+}} \right)}/{H_{gi} P_{gi}^{\max}}. \label{psi_g}
\end{align}%
\label{price_Competitive_g}%
\allowdisplaybreaks[0]%
By comparing the prices in \cref{pi_g}, \cref{varphi_g} and \cref{psi_g} and the prices in \cref{lambda_g}, \cref{gamma_g} and \cref{chi_g}, we can see that $\pi_{gi,t}=\lambda_t^*$, $\varphi_{gt}=\gamma_t^*$, and $\psi_{gt}=\chi_t^*$, $\forall i\in\set{I}$, $\forall t\in\set{T}$. 

Similarly, the KKT optimality conditions for each profit-maximizing ES in  \cref{Competitive_e} are:
\allowdisplaybreaks
\begin{subequations}
\begin{align}
    (P_{di,t})\!:\ &c_{di} + \xi_{i,t}^{+}  - \xi_{i,t}^{-} - \eta_{i,t}  /k_i - \pi_{ei,t} = 0 \label{KKT_Competitive_e_a}\\
    (P_{ci,t})\!:\ & - c_{ci} + \nu_{i,t}^{+} - \nu_{i,t}^{-} + \eta_{i,t}   k_i + \pi_{ei,t} = 0 \label{KKT_Competitive_e_b}\\
    (e_{i,t})\!:\ &\beta_{i,t}^{+} - \beta_{i,t}^{-} + \eta_{i,t} - \eta_{i,t + 1} = 0 \label{KKT_Competitive_e_c}\\
    (\alpha _{di,t})\!:\ &c_{di} \mathrm{M}_{pt} + \xi_{i,t}^{+} {\hat \delta }_{di} + \rho_{di,t}^{+} -\rho_{di,t}^{-} - \varphi_{et} = 0 \label{KKT_Competitive_e_d}\\
    (\alpha_{ci,t})\!:\ &c_{ci} \mathrm{M}_{pt} +\nu_{i,t}^{+}{\hat \delta }_{ci}+\rho _{ci,t}^{+} - \rho_{ci,t}^{-} + \varphi_{et}=0 \label{KKT_Competitive_e_e}\\
    (H_{ei,t})\!:\ &- \psi_{et} P_{ei}^{\max}+2\left( {\xi_{i,t}^{+}  + \nu_{i,t}^{+} } \right){f_{\max }^{\prime}}P_{ei}^{\max}/f_0 \nonumber \\
    &+ 2\left( {\beta _{i,t}^{+} + \beta _{i,t}^{-}} \right) \Delta f_{\max} P_{ei}^{\max}/f_0 = 0. \label{KKT_Competitive_e_f}
\end{align}%
\label{KKT_Competitive_e}%
\end{subequations}%
\allowdisplaybreaks[0]%
Using \cref{KKT_Competitive_e}, we can express $\pi_{ei,t}$, $\varphi_{et}$ and $\psi_{et}$ as:
\allowdisplaybreaks
\begin{align}
    \pi_{ei,t} &= c_{di} + \xi _{i,t}^{+} - \xi _{i,t}^{-} - \eta _{i,t} {\rm{/}} k_i \nonumber\\
    &= c_{ci} - \nu_{i,t}^{\rm{+}} + \nu_{i,t}^{-} - \eta_{i,t} k_i \label{pi_e}\\
    \varphi_{et} &= c_{di} \mathrm{M}_{pt} + \xi_{i,t}^{+} {{\hat \delta}_{di}} + \rho _{di,t}^{+}  - \rho _{di,t}^{-} \nonumber\\
    &= -c_{ci} \mathrm{M}_{pt} - \nu_{i,t}^{\rm{+}}{{\hat \delta}_{ci}} - \rho_{ci,t}^{+}  + \rho_{ci,t}^{-} \label{varphi_e} \\
    \psi_{et} &= \frac {2}{f_0}\left[ {(\xi_{i,t}^{+} {\rm{+}} \nu _{i,t}^{+} ){f_{\max}^{\prime}} {\rm{+}} (\beta_{i,t}^{+}  {\rm{+}} \beta_{i,t}^{-} )\Delta {f_{\max}}} \right]{\rm{+}} \frac{\varepsilon_{i,t}}{P_{ei}^{\max}} \label{psi_e}
\end{align}%
\label{price_Competitive_e}%
\allowdisplaybreaks[0]%
By comparing the prices in \cref{pi_e}, \cref{varphi_e} and \cref{psi_e} and the prices in \cref{lambda_d}, \cref{lambda_c}, \cref{gamma_d}, \cref{gamma_c} and \cref{chi_e}, we see that $\pi_{ei,t}=\lambda_t^*$, $\varphi_{et}=\gamma_t^*$, and $\psi_{et}=\chi_t^*$, $\forall i\in\set{I},\ \forall t\in\set{T}$. 

In summary, we  prove that $\pi_{gi,t}=\pi_{ei,t}=\lambda_t^*$,\ $\forall i\in\set{I},\ \forall t\in\set{T}$, $\varphi_{gt}=\varphi_{et}=\gamma_t^*$,  $\forall t\in\set{T}$, and $\psi_{gt}=\psi_{et}=\chi_t^*$, $\forall t\in\set{T}$. Thus, the set of production levels and prices \{$P_{gi,t}^*$, $P_{di,t}^*$, $P_{ci,t}^*$, $\alpha_{gi,t}^*$, $\alpha_{di,t}^*$, $\alpha_{ci,t}^*$, $u_{gi,t}^*$, $H_{ei,t}^*$, $\pi_{gi,t}$, $\varphi_{gt}$, $\psi_{gt}$, $\pi_{ei,t}$, $\varphi_{et}$, $\psi_{et}$, $\forall i\in\set{I}, \forall t\in\set{T}$\}\ constitutes a competitive equilibrium, i.e. by solving \cref{model_QP}, we can obtain the energy, reserve and inertia prices which clear the market and maximize the profit of all producers.
\end{proof}

\section{Network-Constrained Extension}
\label{sec:network_constrained_extension}

In this section we introduce DC power flow constraints into the CC-UC model in \cref{model_2} and demonstrate in Theorem~\ref{Theorem:LMP} that (i) energy prices now take the form locational marginal prices (LMPs) and (ii) that the results of  Propositions~\ref{prop:reserve_price}-\ref{prop:inertia_price} and Theorem~\ref{Theorem:Equilibrium} remain valid.
To simplify notation, we assume that there is at most one generator, one ES and one RES per node. 

The network-constrained case includes two modifications relative to \cref{model_QP}.
First, we replace the power balance constraint in \cref{model_1_i} with the nodal power balance constraint in \cref{model_net_b}. Second, we add dc power flow constraints in \cref{model_net_c} and set the voltage angle at reference node ($i=ref$) to 0, see \cref{model_net_d}.
The resulting network-constrained CC-UC model is given as:
\allowdisplaybreaks
\begin{subequations}
\begin{align}
    &\min_{\substack{\{P_{gi,t}, P_{di,t}, P_{ci,t},u_{gi,t},\alpha _{gi,t}\\  \alpha _{di,t},\alpha _{ci,t},H_{ei,t}\}_{t \in \set{T}, i \in \set{I}} }}\  C_G+C_{ES} \label{model_net_a}\\
    \text{s.t. }& \cref{model_1_g}-\cref{model_1_i}, \cref{model_1_l}, \cref{model_2_f}-\cref{model_2_p}, \cref{model_determined_b}-\cref{model_determined_g}, \cref{model_QP_b}:\nonumber\\ 
    & (\lambda_{i,t}): P_{gi,t} + P_{di,t} - P_{ci,t} + P_{wi,t}- d_{i,t} =  \nonumber \\
    & \hspace{1.1cm} \sum\nolimits_{j \in \mathcal{N}_i } B_{i,j} (\theta_{i,t} {\rm{-}} \theta_{j,t}) \quad \forall{i}\in\set{I},\ \forall{t}\in\set{T} \label{model_net_b}\\
    & (\vartheta _{i,j,t}^{-},\vartheta _{i,j,t}^{+}): -S_{i,j} \le B_{i,j} (\theta_{i,t}-\theta_{j,t}) \le S_{i,j} \nonumber \\
    & \hspace{1.1cm} \forall{i}\in\set{I},\ \forall{j}\in \mathcal{N}_i ,\ \forall{t}\in\set{T} \label{model_net_c}\\
    & (\varpi_{t}): \theta_{ref,t} = 0 \quad \forall{t}\in\set{T}, \label{model_net_d}
\end{align}%
\label{model_network}%
\end{subequations}%
\allowdisplaybreaks[0]%
where $\mathcal{N}_i $ is the set of nodes that are connected to node $i$, $\theta_{i,t}$ is the voltage angle of node $i$ at time $t$, $B_{i,j}$ is the susceptance of the line between node $i$ and $j$ and $S_{i,j}$ is the thermal capacity of the line between node $i$ and $j$.

\begin{theorem}
\label{Theorem:LMP}
Consider the model in \cref{model_network}. Then (i) energy prices $\lambda_{t}$ from Proposition~\ref{prop:active_power_price} become LMPs $\lambda_{i,t}$ given as:
\allowdisplaybreaks
\begin{align}
    \lambda_{i,t} &= c_{1i} + 2c_{2i}(P_{gi,t} + \mathrm{M}_{pt} \alpha _{gi,t}) + \mu _{i,t}^{+}  - \mu _{i,t}^{-} \label{lambda_i_g}\\
    \lambda_{i,t} &= c_{di} + \xi _{i,t}^{+}  - \xi _{i,t}^{-}  - \eta _{i,t} {\rm{/}} k_i \label{lambda_i_d}\\
    \lambda_{i,t} &= c_{ci} - \nu _{i,t}^{\rm{ + }} + \nu _{i,t}^{-} - \eta _{i,t}   k_i, \label{lambda_i_c}
\end{align}%
\allowdisplaybreaks[0]%
and (ii) the results of Propositions~\ref{prop:reserve_price}-\ref{prop:inertia_price} and Theorem~\ref{Theorem:Equilibrium} hold.
\end{theorem}

\begin{proof}
The KKT conditions of \cref{model_network} are:
\allowdisplaybreaks
\begin{subequations}
\begin{align}
    & \text{\cref{KKT_d}-\cref{KKT_i}}  \nonumber \\ 
    (P_{gi,t})\!: &c_{1i} {\rm{+}} 2c_{2i} \left( {P_{gi,t} {\rm{+}} \mathrm{M}_{pt} \alpha _{gi,t}} \right){\rm{+}} \mu _{i,t}^{+}  {\rm{-}} \tau_{gi,t}^{-} {\rm{-}} \lambda_{i,t} {\rm{=}} 0 \label{KKT_net_a}\\
     (P_{di,t})\!: &c_{di} + \xi _{i,t}^{+}  - \xi _{i,t}^{-} - \eta_{i,t}  /k_i - \lambda_{i,t} = 0 \label{KKT_net_b}\\
     (P_{ci,t})\!: & - c_{ci} + \nu _{i,t}^{+} - \nu _{i,t}^{-} + \eta_{i,t}   k_i + \lambda_{i,t} = 0 \label{KKT_net_c}\\
    (\theta_{i,t})\!: &\sum\nolimits_{j \in \mathcal{N}_i } B_{i,j} (\lambda_{i,t} {\rm{-}} \lambda _{j,t} {\rm{+}} \vartheta_{i,j,t}^{+}  {\rm{-}} \vartheta_{j,i,t}^{+}  {\rm{-}} \vartheta_{i,j,t}^{-}  {\rm{+}} \vartheta_{j,i,t}^{-} ) \nonumber  \\
    & + \varpi_t  = 0 \quad i = ref \label{KKT_net_d}\\
    (\theta_{i,t})\!: &\sum\nolimits_{j \in \mathcal{N}_i } B_{i,j} (\lambda_{i,t} {\rm{-}} \lambda _{j,t} {\rm{+}} \vartheta_{i,j,t}^{+}  {\rm{-}} \vartheta_{j,i,t}^{+}  {\rm{-}} \vartheta_{i,j,t}^{-}  {\rm{+}} \vartheta_{j,i,t}^{-} )  \nonumber  \\
    & {\rm{=}} 0\ i \ne ref. \label{KKT_net_e}
\end{align}%
\label{KKT_net}%
\end{subequations}%
\allowdisplaybreaks[0]%
LMPs in \cref{lambda_i_g}-\cref{lambda_i_c} can be obtained directly from \cref{KKT_net_a}-\cref{KKT_net_c}, which proves (i).

The KKT conditions associated with reserve and inertia prices $\gamma_t$ and $\chi_t$ are identical for \cref{model_network} and \cref{model_QP} and, thus, Propositions \ref{prop:reserve_price}-\ref{prop:inertia_price} remain valid for \cref{model_network}. 
Further, since the right-hand sides of \cref{lambda_i_g}, \cref{lambda_i_d}, and \cref{lambda_i_c} are identical to \cref{pi_g} and \cref{pi_e}, respectively, and \cref{gamma_t}, \cref{chi_t} remain unchanged, duals \{$\lambda_{i,t}^*\ \forall i\in\set{I}$, $\gamma_t^*$, $\chi_t^*$, $\forall t\in\set{T}$\} and quantities \{$P_{gi,t}^*$, $P_{di,t}^*$, $P_{ci,t}^*$, $\alpha_{gi,t}^*$, $\alpha_{di,t}^*$, $\alpha_{ci,t}^*$, $u_{gi,t}^*$, $H_{ei,t}^*$, $\forall i\in\set{I}, \forall t\in\set{T}$\}  obtained from an optimal solution of \cref{model_network} also yield a competitive equilibrium. This concludes the proof of (ii).
\end{proof}

\section{Case Study}

\subsection{Illustrative Example}
\label{sec:numerical_experiments_toy_system}

We consider an illustrative example with four generators, two ES, one wind farm and the system-wide load as given in in Tables \ref{tab:G parameters}, \ref{tab:ES parameters} and \ref{tab:other parameters}. Generator $\text{G}_i,\ i{\rm{=}}1,2,3,4$ denotes the generator at node $i$ and $\text{ES}_i,\ i{\rm{=}}1,2$ denotes the ES at node $i$. The load and wind power profiles are given in Fig.~\ref{fig:small case load and wind}. For simplicity, we assume that the distribution parameters of $\omega_{pi,t}$ and $\omega_{hi,t}$ are time-invariant.
Note that in this illustrative case study, the charging and discharging prices of ES (i.e., $C_d$ and $C_c$ in Table \ref{tab:ES parameters}) are set to relatively low values, which ensures that the effect of adding ES as an inertia provider is obvious.
\begin{table}[!t]
  \centering
  \caption{Generator parameters}
    \begin{tabular}{cccccc}
    \toprule
    \multirow{2}[0]{*}{No.} & $H_g$ & $P_{g}^{\max}/P_{g}^{\min}$  & $C_0$ & $C_1$    & $C_2$  \\
    & (s)   & (MW)  &  ($\$$)  & ($\$$/MWh) & ($\$$/MWh$^2$) \\
    \midrule
    G1    & 6     & 10/1  & 10    & 5     & 0.001 \\
    G2    & 6     & 10/1  & 50    & 12    & 0.003 \\
    G3    & 6     & 10/1  & 80    & 15    & 0.005 \\
    G4    & 10    & 10/1  & 150   & 30    & 0.006 \\
    \bottomrule
    \end{tabular}%
  \label{tab:G parameters}
\end{table}

\begin{table}[!t]
  \centering
  \caption{ES parameters}
    \resizebox{\linewidth}{!}{%
    \begin{tabular}{cccccc}
    \toprule
    \multirow{2}[0]{*}{No.} & $H_e^{\max}$ & $P_d^{\max}$/$P_c^{\max}$  & $E^{\max}$/$E^{\min}$  & $C_d$  & $C_c$  \\
    & (s)   & (MW)  &  (MWh)  & (\$/MWh) & (\$/MWh) \\
    \midrule
    ES1   & 11  & 10 / 5 & 10 / 0.5   & 5   & 10 \\
    ES2   & 11  & 10 / 5 & 10 / 0.5   & 7   & 12 \\
    \bottomrule
    \end{tabular}
    }
  \label{tab:ES parameters}
\end{table}

\begin{table}[t!]
  \centering
  \caption{Auxiliary model parameters}
    \begin{tabular}{cc|cc|cc}
    \toprule
    Parameter     & Value      & Parameter          & Value        & Parameter  & Value\\
    \midrule
    $H_{\min}$    & 3.5 s      & $f_{\max}^{'}$     & 0.5 Hz/s  & $k_i,i{\rm{=}}1,2$              & 0.9 \\
    $P_{sys}$     & 80 MW      & $\Delta f_{\max}$  & 0.55 Hz  & $\epsilon_{gi},i{\rm{=}}1...4$  & 0.05 \\
    $P_w^{\max}$         & 20 MW      & $\Sigma_{pt}$, $\sigma_{hi,t}$  & 1   & $\epsilon_{di},i{\rm{=}}1,2$    & 0.05 \\
    $f_0$         & 50 Hz      & $\mathrm{M}_{pt}$, $\mu_{hi,t}$     & 0.5   & $\epsilon_{ci},i{\rm{=}}1,2$    & 0.05 \\
    \bottomrule
    \end{tabular}
  \label{tab:other parameters}
\end{table}

\begin{figure}[!t]
    \centering
    \includegraphics[width=1\linewidth]{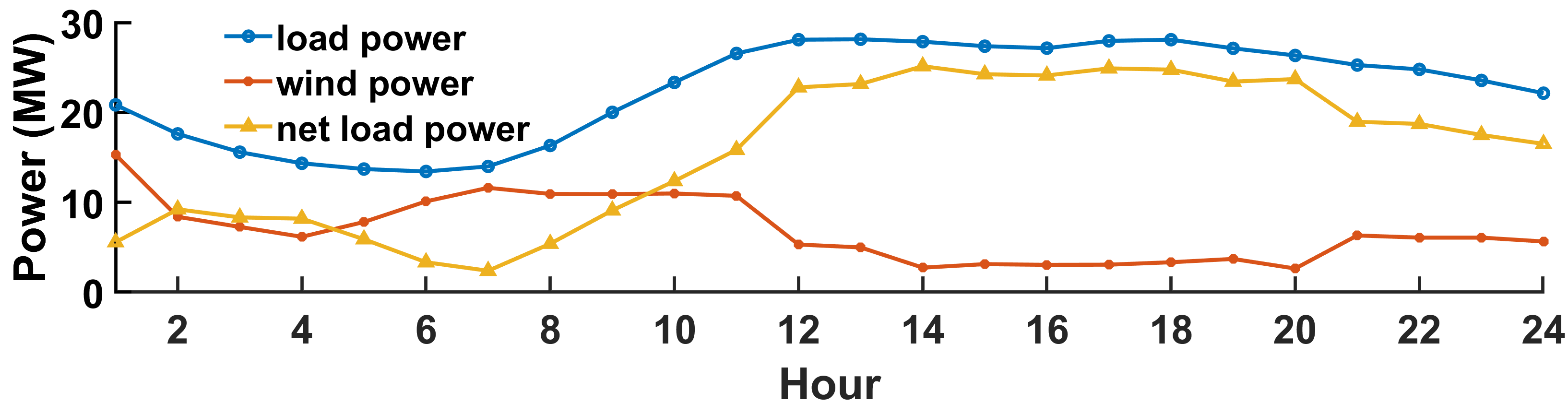}
    \caption{Load, wind power and net load profiles used in the illustrative example. }
    \label{fig:small case load and wind}
\end{figure}

Using the model in \cref{model_QP}, we compare its performance in six cases, which are summarized in Table \ref{tab:cases}, where Case 1 represents the status-quo and Case 6 allows for the energy, reserve and inertia provision by all committed generators, ES and RES. We also evaluate the effect of optimizing ES inertia constant $H_e$ (as in Cases 4-6 where $0 \le H_{ei}\le \unit[11]{s},\ i{\rm{=}}1,2$) over treating it as a constant (as in Case 3 where $H_{ei}=\unit[8]{s},\ i{\rm{=}}1,2$). Case 6 is modeled as \cref{model_QP}. For Cases 1-5, the model is derived from \cref{model_QP} by enforcing some variables to be zero. Specifically, in Case 1, $\alpha _{di,t},\ \alpha _{ci,t},\ H_{ei,t},\ H_{wi,t} = 0,\ \forall t\in\set{T}\ \forall i\in\set{I}$, in Case 2, $H_{ei,t},\ H_{wi,t} = 0,\ \forall t\in\set{T}\ \forall i\in\set{I}$, in Cases 3-4, $\alpha _{di,t},\ \alpha _{ci,t},\ H_{wi,t}=0\ \forall t\in\set{T}\ \forall i\in\set{I}$, and in Case 5, $H_{wi,t}=0,\ \forall t\in\set{T}\ \forall i\in\set{I}$.

\begin{table}[!t]
  \centering
  \caption{Overview of Studied Cases}
    \begin{tabular}{cccccccc}
    \toprule
    \multicolumn{2}{p{5.25em}}{Case} & 1 & 2  & 3  & 4  & 5  & 6 \\
    \midrule
    \multirow{3}[0]{*}{G} & Energy & \checkmark     &\checkmark     & \checkmark     & \checkmark    & \checkmark    & \checkmark \\
          & Reserve     & \checkmark    & \checkmark     & \checkmark     & \checkmark    & \checkmark    & \checkmark \\
          & Inertia     & \checkmark    & \checkmark     & \checkmark     & \checkmark    & \checkmark    & \checkmark\\
    \midrule
    \multirow{3}[0]{*}{ES} & Energy & \checkmark    & \checkmark     &\checkmark    & \checkmark    & \checkmark     & \checkmark \\
          & Reserve &   & \checkmark    &       &        & \checkmark     & \checkmark \\
          & Inertia &       &           & \checkmark (Con) & \checkmark(Var) & \checkmark(Var)   & \checkmark(Var) \\
    \midrule
    \multirow{2}[0]{*}{W} & Energy & \checkmark    & \checkmark   & \checkmark        & \checkmark    & \checkmark   & \checkmark \\
          & Inertia &       &       &       &       &       &\checkmark  \\
    \bottomrule
    \end{tabular}
  \label{tab:cases}
\end{table}

Fig.~\ref{fig:small case prices}-\ref{fig:small case UC} and Table~\ref{tab:economic results} summarize the market-clearing outcomes obtained in Cases 1-6 in terms of the energy, reserve and inertia prices, the commitment decisions of conventional units and the economic performance of resources.  

\begin{figure}[!t]
    \centering
    \includegraphics[width=1\linewidth]{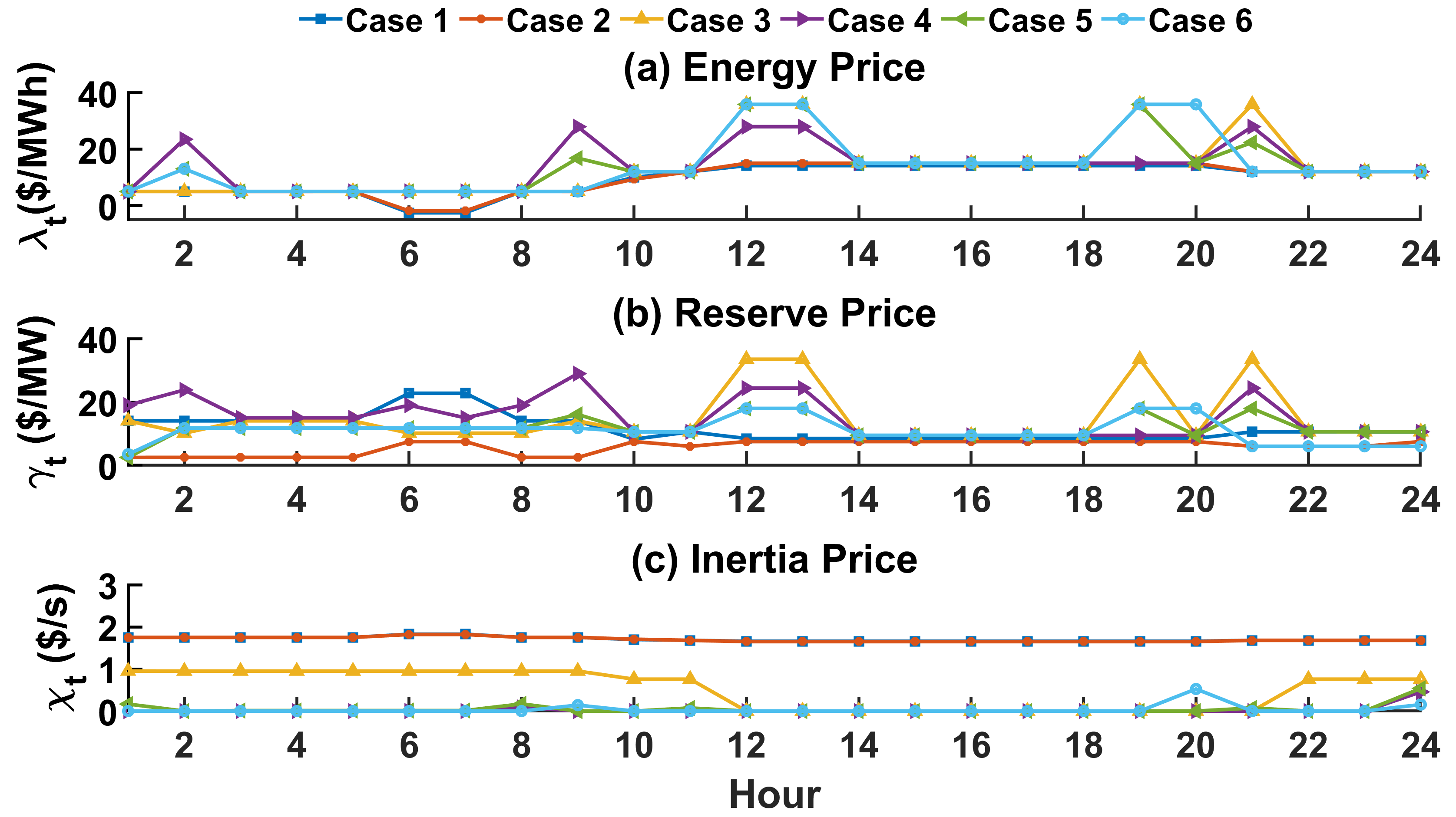}
    \caption{Prices for providing energy, reserve and inertia in Cases 1-6 of the illustrative example.}
    \label{fig:small case prices}
\end{figure}

\begin{figure}[!t]
    \centering
    \includegraphics[width=1\linewidth]{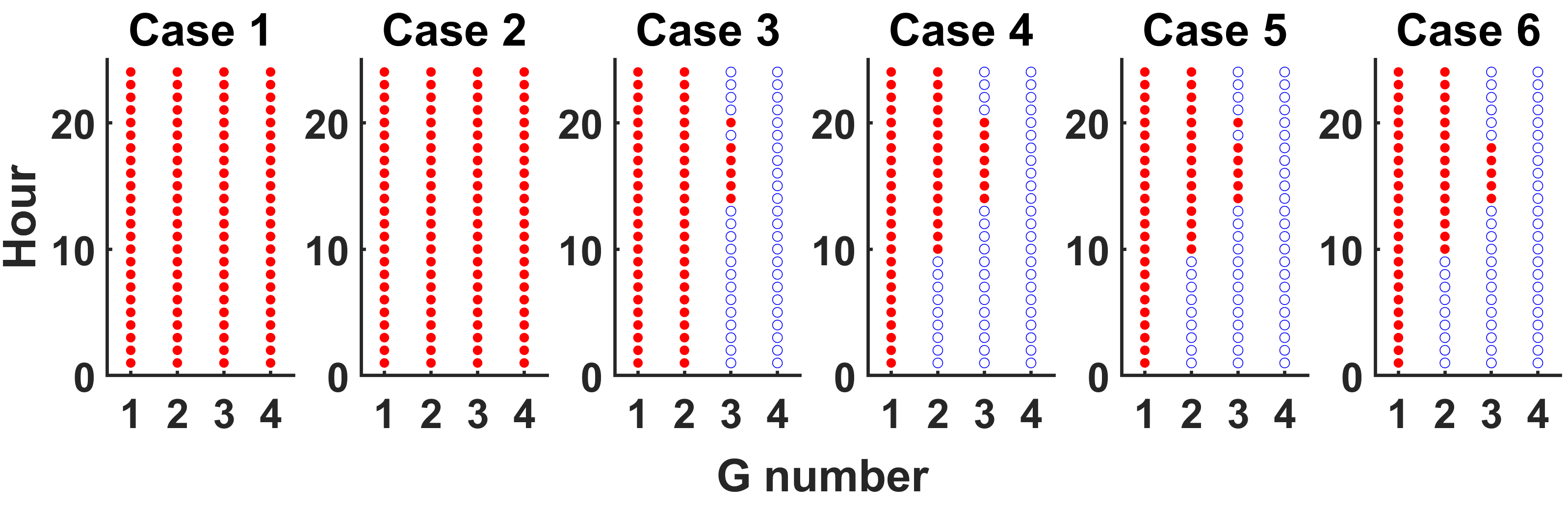}
    \caption{Unit commitment decisions in Cases 1-6 of the illustrative example, where solid red circles indicate committed generators ($u_{gi,t}=1$) and hollow blue circles indicate offline generators ($u_{gi,t}=0$).}
    \label{fig:small case UC}
\end{figure}

Fig.~\ref{fig:small case prices}(a) shows that the average energy prices in Cases 1 and 2 are lower than in Cases 3-6. This is because the former cases have more restrictive operational constraints where the inertia provision falls on generators, leading to the commitment of all generators and out-of-merit order dispatch. For example, in Cases 1-2 the capacity of the most expensive generator (G$_4$) is not necessary to meet the system peak load (\unit[28.14]{MW}, see Fig.~\ref{fig:small case load and wind}) because cheaper generators G$_1$, G$_2$ and G$_3$ can produce up to \unit[30]{MW}, but it is committed to meet inertia requirements. As a result, G$_1$, G$_2$ and G$_3$ with low marginal costs are under-loaded and set up the system energy price. Therefore, when ES and RES in Cases 3-6 provide virtual inertia, G$_4$ is not committed (see Fig.~\ref{fig:small case UC}), which allows for a dispatch point closer to the merit order. 
As a result of the out-of-merit order dispatch, Cases 1 and 2 also tend to produce lower reserve prices than Cases 3-6 (see Fig.~\ref{fig:small case prices}(b)). In other words, enabling ES and RES provide reserve and inertia leads to reduced commitments in Cases 3-6 (see Fig.~\ref{fig:small case UC}), which makes less conventional generation capacity available for reserve provision, leading to reserve price increases.  
However, as the number and flexibility of reserve and inertia providers increases in Cases 5 and 6,  the resulting reserve prices become closer to Cases 1 and 2. 
As expected, Cases 1 and 2 lead to the greatest inertia prices because in these cases only conventional generators are eligible to provide inertia, which requires the commitment of all four generators. 
As more resources become available to provide inertia and reserve in Cases 3-6, the inertia prices gradually decline (see Fig.~\ref{fig:small case prices}(c)). Notably, Cases 4-6 leads to lower inertia prices than Case 3 because in these cases the inertia provision from ES units is optimized rather than based on a fixed value.

In terms of the system-wide performance, introducing inertia and reserve provision from RES and ES reduce the total operating cost, see Table~\ref{tab:economic results}. This trend is observed uniformly from the most restrictive (Case 1) to the most relaxed (Case 6) instances. Meanwhile, the overall profit of all generators, RES and ES increases in Cases 3-6 relative to Cases 1-2 due to greater energy and reserve prices as discussed above.

Table~\ref{tab:economic results} itemizes the total cost, revenue and profit of all resources in each case. 
We observe that the objective value, i.e., the total system cost, decreases from Case 1 to Case 6. This is because the scheduling of the power system becomes more flexible when ES and the wind farm can provide reserve and inertia.  To be more specific, in Case 6, the reserve provided by ES account for 27.29 \% of the total reserve, while the virtual inertia provided by ES and wind farms account for 45.26 \% and 16.36 \% of the total inertia. Comparing with Case 1, where generators provided all the reserve and inertia, the total operating cost in Case 6 decreased by 56.64\%, which shows the significance of ES and wind farms participating the reserve and inertia market in a RES-rich system.

\begin{table}[t!]
  \centering
  \caption{Daily cost, revenue and profit in Cases 1-6 (in \$).}
  \resizebox{\linewidth}{!}{%
    \begin{tabular}{cccccccc}
    \toprule
    \multicolumn{2}{p{3.35em}}{Case} & 1     & 2     & 3     & 4     & 5     & 6 \\
    \midrule
    \multicolumn{2}{r}{Total cost} & 10835.1  & 10825.6  & 5324.4  & 4826.9  & 4775.8  & 4697.8  \\
    \midrule
    \multirow{3}[0]{*}{G} & Revenue & 5162.6  & 5213.4  & 7293.2  & 6145.4  & 6523.8  & 6500.0  \\
          & Cost  & 10720.0  & 10734.6  & 5130.8  & 4705.2  & 4603.1  & 4440.3  \\
          & Profit & -5557.3  & -5521.2  & 2162.3  & 1440.2  & 1920.8  & 2059.7  \\
    \midrule
    \multirow{3}[0]{*}{ES} & Revenue & 169.7  & 159.9  & 477.0  & 269.2  & 430.1  & 588.2  \\
          & Cost  & 115.1  & 91.1  & 193.6  & 121.7  & 172.8  & 257.4  \\
          & Profit & 54.5  & 68.8  & 283.5  & 147.4  & 257.3  & 330.8  \\
    \midrule
    W     & Revenue & 1260.3  & 1294.4  & 1917.1  & 2114.4  & 2028.4  & 1889.7  \\
    \bottomrule
    \end{tabular}
    }
  \label{tab:economic results}
\end{table}

\subsection{Numerical Experiments on the IEEE 118-Bus System}
\label{sec:numerical_experiments_large_system}

This section presents results obtained with the network-constrained extension in \cref{model_network} using the modified IEEE 118-bus system from \cite{zimmerman2010matpower} with added 11 wind farms and 11 ES units. It is assumed that ES and RES units are co-located at nodes: 3, 8, 11, 20, 24, 26, 31, 38, 43, 49, 53. 
The system-wide load and wind power profiles are shown in Fig.~\ref{fig:large case prices}(a) and distributed among 91 buses in the original system as described in \cite{zimmerman2010matpower}. We also set the minimum inertia constant requirement ($H_{\min}$) to a typical value of \unit[3.3]{s} \cite{fernandez2020review}.
This inertia requirement must then be met by ES, RES and 54 traditional generators in the system with $H_{gi}=\unit[3.5]{s} \ i=1,\cdots,10$ (hydro power generator), $H_{gi}=\unit[4]{s} \ i=11,\cdots,35$ (coal, oil or nuclear power generator), and $H_{gi}=\unit[5]{s} \ i=36,\cdots,54$ (gas power generator) \cite{fernandez2020review}. Finally, the inertia constant of each ES unit is constrained as $H_e \in [0, \unit[11]{s}]$.

The scope of this section is limited to Cases 1 and 6 as defined in Table~\ref{tab:cases} because it allows for comparing the status quo and the most advanced case where all available conventional, ES and RES resource compete for the provision of energy, reserve and inertia services. To be more specific, in Case 6, the reserve provided by ES account for 15.67 \% of the total reserve, while the virtual inertia provided by ES and wind farms account for 34.52 \% and 12.98 \% of the total inertia. Therefore, the total operating cost reduces by 34.9\% in Case 6 (\$18841.59) relative to Case 1 (\$28956.03).

Inspecting LMPs at the reference node (Bus 1) reveals that the outcomes are very similar in Cases 1 and 6 (see  Fig.~\ref{fig:large case prices}(b)). The modest difference in LMP between these cases during 12:00-20:00 is caused by the out-of-merit order dispatch in Case 1, which we have already discussed in the illustrative example in Section~\ref{sec:numerical_experiments_toy_system}.
Fig.~\ref{fig:LMP and its distribution}(a) and (b) demonstrate that in both Cases 1 and 6, LMPs at different nodes are the same most of the time, except for two periods. One period includes 01:00, 06:00 and 07:00 and the negative LMPs at nodes 1-33 (where most of the wind farms are concentrated) are caused by excessive wind power, which requires ES charging, see Fig.~\ref{fig:LMP and its distribution}(c). The other period is 12:00-20:00, and the LMP difference is caused by congestion on lines 7-9, 41, 119 and 152 due to high net load, see Fig.~\ref{fig:LMP and its distribution}(d). 

On the other hand, reserve prices reported in Fig.~\ref{fig:large case prices}(b) are systematically greater for Case 1 relative to Case 6 because of fewer available resources providing reserve. For example, Fig.~\ref{fig:large case reserve}(a) and (b) compare the reserve provision in Cases 1 and 6, from which it follows that in Case 6 reserve is provided exclusively by conventional generators, while in Case 6 reserve burden is provide by ES units, which have lower marginal costs. 

Fig.~\ref{fig:large case prices}(d) shows that inertia prices are zero at all times in Case 1, which indicates that the committed generators in Case 1 involuntarily provide more inertia than the least inertia requirement ($H_{\min}$) and this constraint is never binding. Unlike Case 1, Case 6 leads to non-zero inertia prices at time periods 01:00, 08:00 and 10:00. After inspecting binding constraints which set the inertia price, it is observed that in all instances ES units are the marginal inertia providers.

The prices reported above make it possible to analyze the daily revenues of generators, RES and ES from providing  energy, reserve and inertia services. Fig.~\ref{fig:large case revenue} compares these revenues for Case 1 and Case 6. In both cases, conventional generators collect their revenue mostly from energy services. However, the total revenue of ES units shifts from energy arbitrage in Case 1 to a combination of energy, reserve and inertia services in Case 6. More specifically, Fig.~\ref{fig:large case revenue} demonstrates that the ES revenue from reserve and inertia services is collected when these resources set respective prices. The revenue of RES is not discussed here since RES is a price taker in this market, i.e. the marginal resources of energy and inertia are always traditional generators or ES.

\begin{figure}[!t]
    \centering
    \includegraphics[width=1\linewidth]{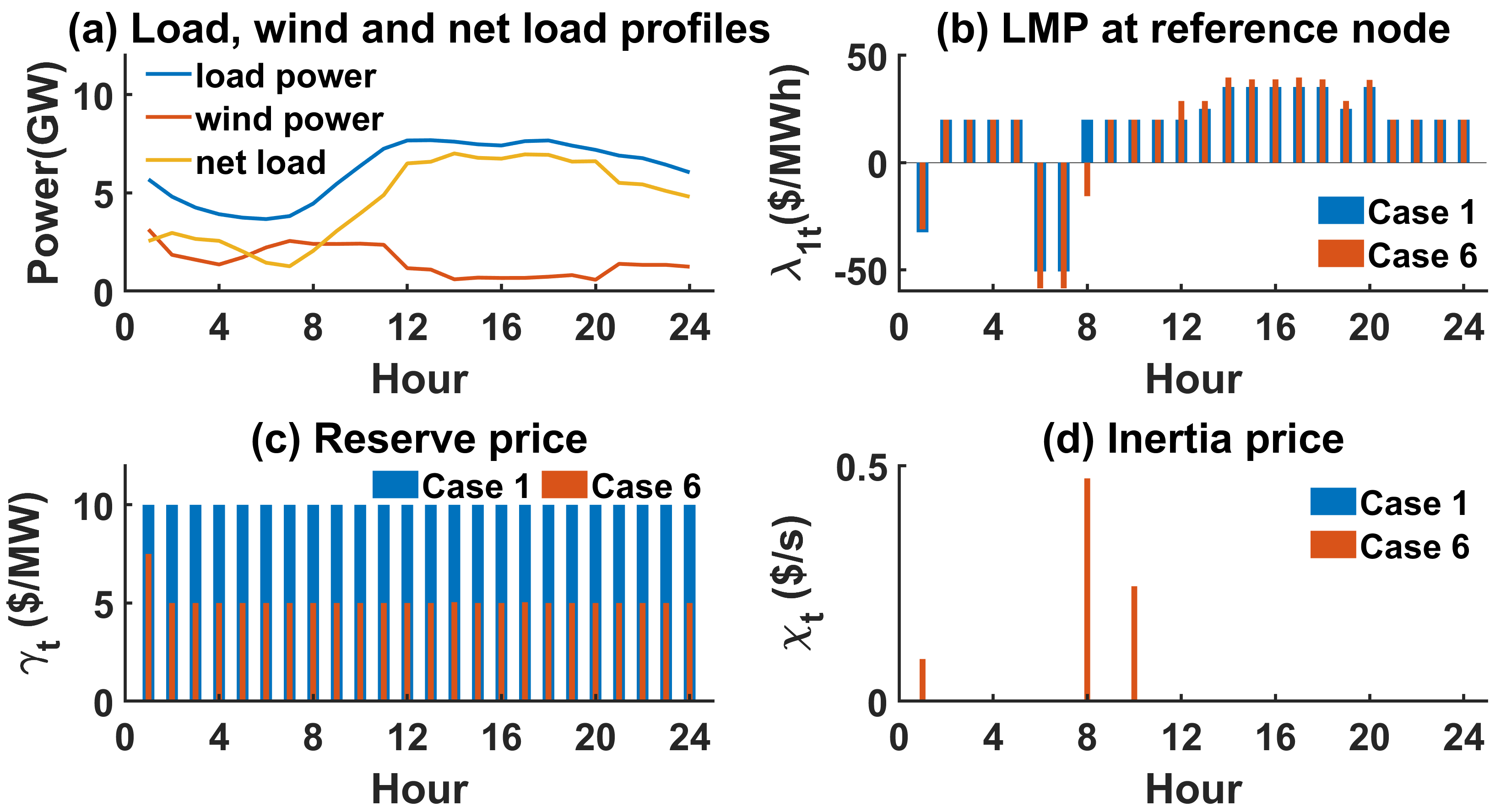}
    \caption{Original load and wind power data and resulting prices in Cases 1 and 6 for the 118-bus IEEE system}
    \label{fig:large case prices}
\end{figure}

\begin{figure}[!t]
    \centering
    \includegraphics[width=1\linewidth]{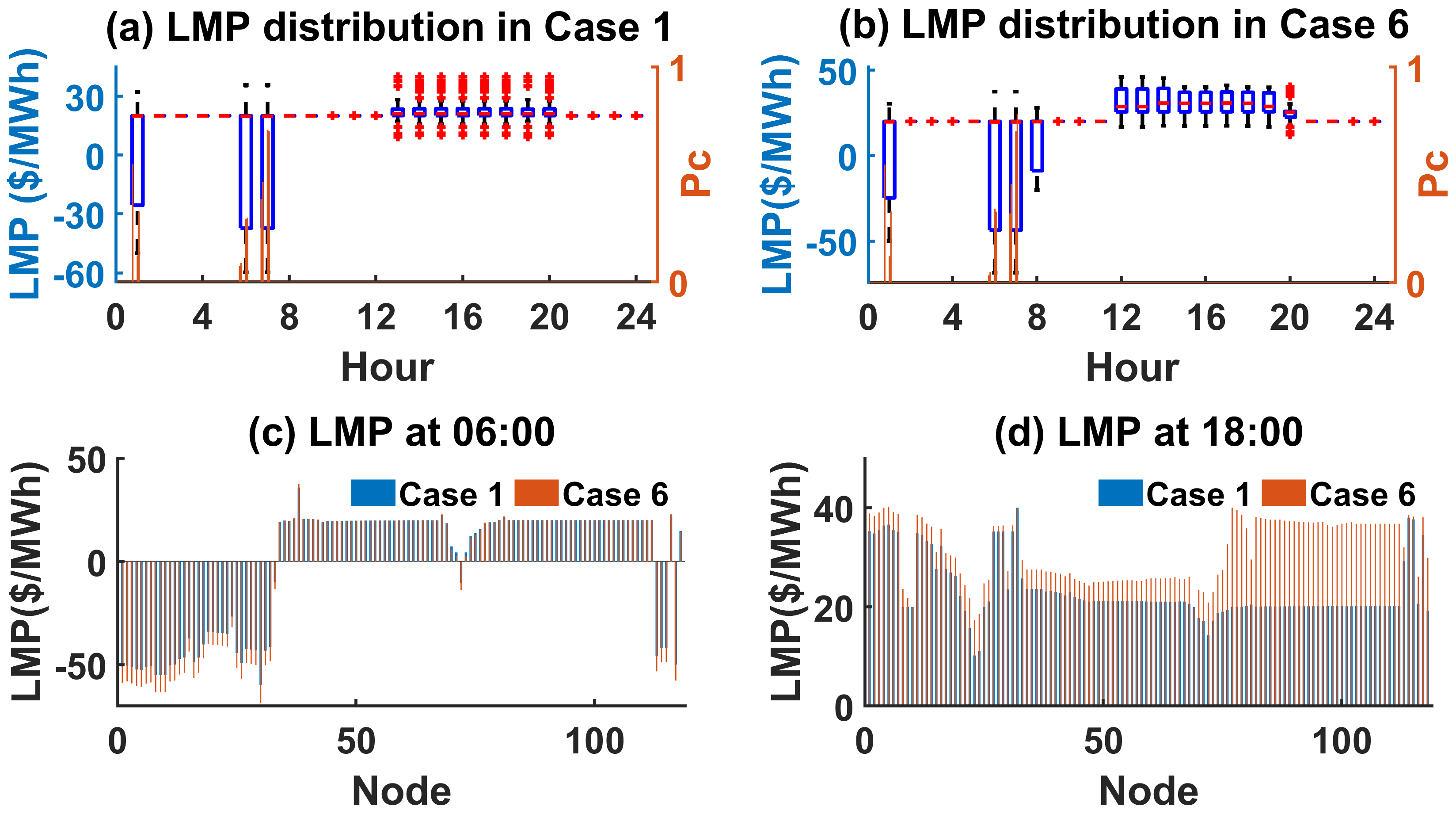}
    \caption{LMPs at selected time periods and their distributions at 118 nodes.}
    \label{fig:LMP and its distribution}
\end{figure}

\begin{figure}[!t]
    \centering
    \includegraphics[width=1\linewidth]{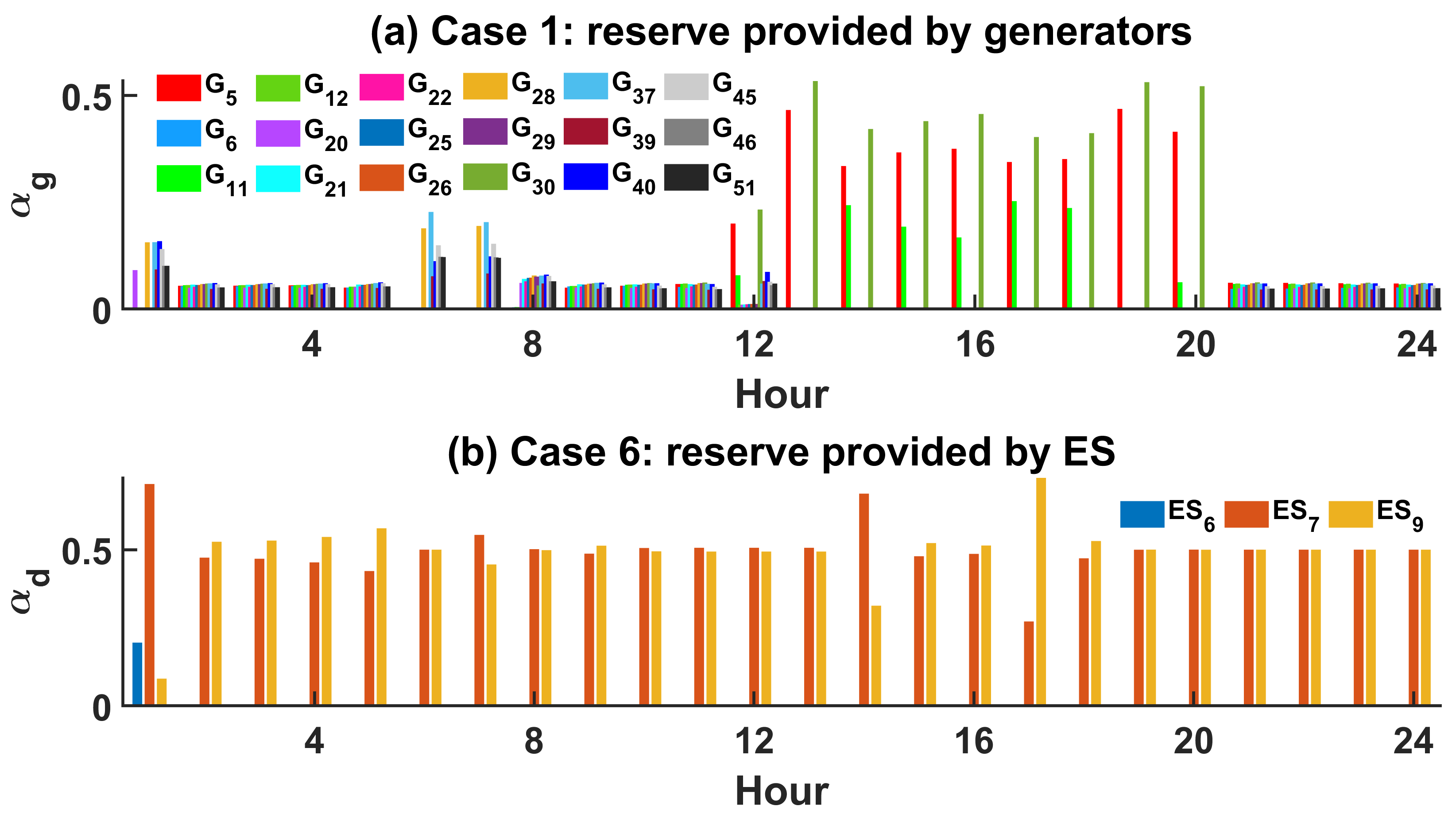}
    \caption{Reserve provision decisions in Cases 1 and 6 for the 118-bus IEEE system}
    \label{fig:large case reserve}
\end{figure}

\begin{figure}[!t]
    \centering
    \includegraphics[width=1\linewidth]{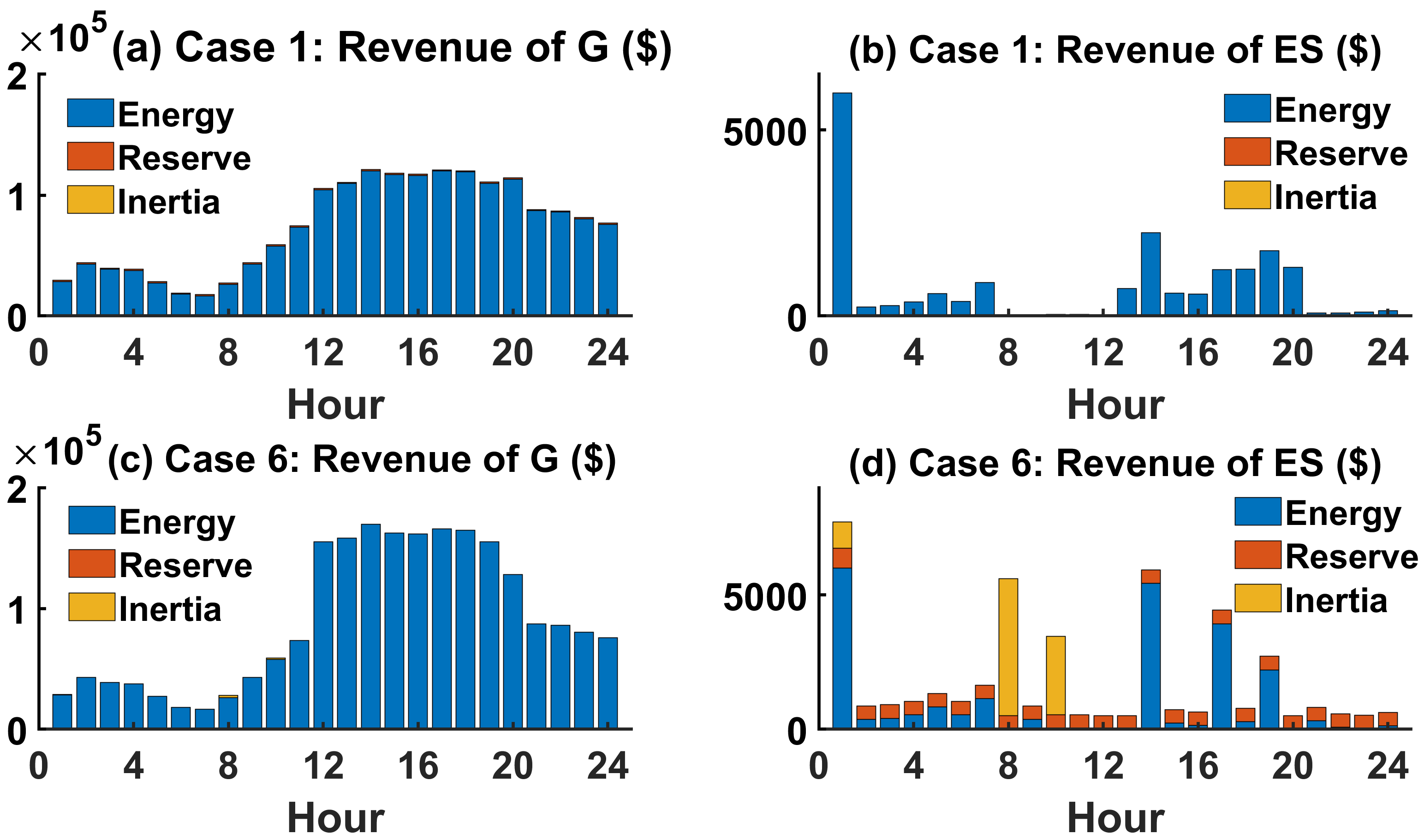}
    \caption{Daily revenue from providing energy, reserve and inertia  for  the 118-bus IEEE system.}
    \label{fig:large case revenue}
\end{figure}

\subsection{Simulation Environment}
All simulations were carried out in Julia v1.5. The MIQP and QP problems were solved using the Gurobi \cite{gurobi} and Ipopt \cite{Wachter2006On} solvers. 
All experiments were performed on a standard PC workstation with an Intel i9 processor and 16 GB RAM. The solving time for each instance in Section~\ref{sec:numerical_experiments_toy_system} was less than 10 seconds, while every instance in Section~\ref{sec:numerical_experiments_large_system} was solved in less than 1 minute. Since the quadratic objective of our model is convex, all problems were solved to global optimality with a duality gap of $<0.01\%$.

\section{Conclusion}
\label{sec:conclusion}

This paper designs a stochastic electricity market to price energy, reserve and inertia provision in renewable-rich power systems. We prove that the resulting market outcomes are efficient and constitute a competitive equilibrium, i.e., they clear the market, minimize the cost and provide no incentive for market participants to deviate from the market outcomes. Numerical experiments are carried out for a 4-generator illustrative system and a modified IEEE 118-bus system, and the results demonstrate the effect of different synchronous and non-synchronous inertia technologies on the dispatch decisions and the resulting energy, reserve and inertia prices. Possible future works include comparing the MIQP relaxation via fixed binaries with alternative methods. This may include ensuring long-term generator profits as discussed in \cite{byers2021long}.

\bibliographystyle{IEEEtran}
\bibliography{literature}

\end{document}